%% file: main.tex
\documentclass[11pt,letterpaper]{article}

\usepackage[margin=1in]{geometry}

\usepackage{lineno}

\usepackage{cite}
\usepackage{hyperref}
\usepackage{colortbl}
\usepackage{paralist}
\usepackage{pdfpages}
\usepackage{enumitem}
\usepackage{float}
\usepackage{booktabs}
\usepackage[override]{cmtt}
\usepackage{color,xcolor}
\usepackage{graphicx,eso-pic}
\usepackage{boxedminipage}
\usepackage{url}
\usepackage{caption}
\usepackage{subcaption}
\usepackage{xspace}
\usepackage{multirow}
\usepackage{amsmath,amsthm,amstext,amssymb,amsfonts,latexsym}
\usepackage{wrapfig}
\usepackage[linesnumbered,ruled,vlined]{algorithm2e}
\usepackage[noend]{algpseudocode}
\usepackage{algorithmicx}
\usepackage{epstopdf}
\usepackage{mdframed}
\usepackage{cases}
\usepackage[capitalize]{cleveref}
\usepackage[compact]{titlesec}
\usepackage{bm}

\input{macro.tex}
\begin{document}

\title{Game-Theoretically Secure Distributed Protocols for Fair Allocation in Coalitional Games}

\author{T-H. Hubert Chan\thanks{Department of Computer Science, the University of Hong Kong.} \and Qipeng Kuang\samethanks \and Quan Xue\samethanks}
\date{}

\begin{titlepage}
    \maketitle

    \input{abstract}


    \thispagestyle{empty}
\end{titlepage}

\input{introduction}

\input{related}

\input{preliminary}

\input{algorithm}

\input{stability}

\input{expectation}

\input{adversary}

\input{experiment}

\input{conclusion}

\bibliographystyle{plain}
\bibliography{ref}

\end{document}

%% file: macro.tex
\newcommand{\quan}[1]{\refstepcounter{note}$\ll${\sf Quan's
		Comment~\thenote:} {\sf \textcolor{red}{#1}}$\gg$\marginpar{\tiny\bf XQ~\thenote}}

\newcommand{\yuetodo}[1]{{\large\color{green}[Yue todo: #1]}}

\definecolor{yue}{rgb}{0.7, 0, 0}

\renewcommand{\yuetodo}[1]{}

\newcommand{\mcal}[1]{\ensuremath{\mathcal {#1}}}

\definecolor{darkgreen}{rgb}{0,0.5,0}
\definecolor{lightblue}{RGB}{0,176,240}
\definecolor{darkblue}{RGB}{0,112,192}
\definecolor{lightpurple}{RGB}{124, 66, 168}
\definecolor{grey}{RGB}{139, 137, 137}
\definecolor{maroon}{RGB}{178, 34, 34}
\definecolor{green}{RGB}{34, 139, 34}
\definecolor{types}{RGB}{72, 61, 139}
\definecolor{gold}{rgb}{0.8, 0.33, 0.0}

\definecolor{darkgray}{gray}{0.3}

\definecolor{darkred}{rgb}{0.5, 0, 0}
\definecolor{darkgreen}{rgb}{0, 0.5, 0}
\definecolor{darkblue}{rgb}{0,0,0.5}

\newcommand\markx[2]{}

\newcommand{\R}{\mathbb{R}}

\newcommand{\aux}{{\sf aux}}

\newcommand{\negl}{{\sf negl}}

\newcommand{\ignore}[1]{}

\newcounter{task}

\newtheorem{theorem}{Theorem}[section]

\newtheorem{corollary}[theorem]{Corollary}
\newtheorem{fact}[theorem]{Fact}
\newtheorem{lemma}[theorem]{Lemma}
\newtheorem{claim}[theorem]{Claim}
\newtheorem{definition}[theorem]{Definition}
\theoremstyle{remark}
\newtheorem{remark}[theorem]{Remark}
\newcounter{cnt:challenge}

\newcommand*\samethanks[1][\value{footnote}]{\footnotemark[#1]}

\newcommand{\E}{\mathbf{E}}

\SetKwComment{Comment}{//}{}

\newcommand{\all}{\mathcal N}
\renewcommand{\Pr}[1]{\mathbf{Pr} \left[ {#1} \right]}
\newcommand{\phii}{\phi_{i^*}}

\newcommand{\Adv}{\mathsf{Adv}}

%% file: abstract.tex
\begin{abstract}
We consider game-theoretically secure distributed protocols for coalition games that approximate the Shapley value with small multiplicative error. Since all known existing approximation algorithms for the Shapley value are randomized, it is a challenge to design efficient distributed protocols among mutually distrusted players when there is no central authority to generate unbiased randomness. The game-theoretic notion of maximin security has been proposed to offer guarantees to an honest player's reward even if all other players are susceptible to an adversary.

Permutation sampling is often used in approximation algorithms for the Shapley value. A previous work in 1994 by Zlotkin et al. proposed a simple constant-round distributed permutation generation protocol based on commitment scheme, but it is vulnerable to rushing attacks. The protocol, however, can detect such attacks.

In this work, we model the limited resources of an adversary by a violation budget that determines how many times it can perform such detectable attacks. Therefore, by repeating the number of permutation samples, an honest player's reward can be guaranteed to be close to its Shapley value. We explore both high probability and expected maximin security. We obtain an upper bound on the number of permutation samples for high probability maximin security, even with an unknown violation budget. Furthermore, we establish a matching lower bound for the weaker notion of expected maximin security in specific permutation generation protocols. We have also performed experiments on both synthetic and real data to empirically verify our results.

\end{abstract} 

%% file: introduction.tex
\section{Introduction}

A coalitional game (\emph{aka} cooperative game) involves $n$ players in the set
$\all$, and the utility function $v: 2^\all \rightarrow \R_+$ determines the reward $v(S)$ obtained by a coalition $S \subseteq \all$ when players cooperate together.
%This study specifically focuses on monotone games, where the utility function exhibits monotonicity.
An extensively researched problem in this context is the fair allocation of the reward among the $n$ players.

\ignore{
In this work, we focus on the special case of a monotone game, where the utility function is monotone,
to receive the total reward $v(\all)$.  A well-studied problem
is how to allocate this reward ``fairly'' among the
$n$ players.

In this work, we focus on the special case of a \emph{supermodular} (\emph{aka convex}) game~\cite{sv_core}, where the utility function is \emph{supermodular}, in which the best strategy
is for all players to cooperate together to form the \emph{grand coalition}
to receive the total reward $v(\all)$.  A well-studied problem
is how to allocate this reward ``fairly'' among the
$n$ players.}

%\quan{Give some references. Also, related WWW. More natural way. }
The \emph{Shapley value} is an allocation concept in coalitional games,
which has many applications and can provide a fair and consistent framework to assess and distribute rewards among different agents, leading to a more equitable, sustainable, and inclusive system. For instance, it can be used to develop an axiomatic framework for attribution in online advertising~\cite{DBLP:conf/www/BesbesDGIS19}, measure the significance of nodes in the context of network centrality~\cite{DBLP:conf/www/0013T17} and develop a pricing scheme of personal recommendations~\cite{DBLP:conf/www/DuttingHW10}.
Computing the exact Shapley value is \#P-hard~\cite{sv_hardness}, so approximations are used in practice. We focus on the notion of $\epsilon$-\emph{multiplicative error}, where each player receives at least $(1-\epsilon)$ fraction of its Shapley value, with $\epsilon > 0$ as a constant.

All known efficient algorithms for approximating the Shapley value require randomness,
and one approach involves generating uniformly random permutations of all players.
However,  this study focuses on the \emph{distributed setting} where no \emph{trusted authority} provides randomness. Each player has its own local source of randomness. The objective is to design an efficient \emph{protocol} for allocating rewards close to the Shapley value with a small multiplicative error, while accounting for players' incentives to deviate.

%Since in a random permutation, players prefer later positions.
%The sampling methods assume unbiased randomness. Random permutation favors later players, but a uniformly random permutation is fair. This paper aims to achieve computational efficiency via distributed protocols for generating randomness among untrusted players, ensuring each player's reward is close to its Shapley value.

%an approximation protocol involving the own randomness of each player is necessary.
%In distributed settings
%Moreover, a standard approximation algorithm which by the input of the game information gives the output through a black-box calculation might not be applicable in a practical environment such as a distributed network. If there is no trusted third party, whoever runs the approximation algorithm containing randomness might get no trust from others. Therefore, an approximation protocol involving the own randomness of each player is necessary.

For distributed protocols, the game-theoretic concept of \emph{maximin security}~\cite{maximin} is introduced to offer security guarantees from the perspective of an \emph{honest} player,
where all other players may be \emph{susceptible} to an \emph{adversary}.
%This is much more stringent than the typical assumption that a constant fraction of players are honest.
%because we are in a game theoretic setting in which players will deviate as long as there is a reward incentive. While the precise capabilities of the adversary depend on the specific problem, the general theme is that even under the influence of the adversary, the reward received by an honest player has some guarantee.
For coalitional games, the security guarantee can be (i) the expected reward received by a player has $\epsilon$-multiplicative error, or (ii) the reward achieves this multiplicative error with high probability.

One previous attempt has been made to design a secure distributed protocol for returning an allocation in a coalitional game. Zlotkin et al.~\cite{sv_protocol} proposed $\mathsf{NaivePerm}$,
%, a random permutation protocol similar to Blum's protocol~\cite{tele_coin_flipping}.
in which players commit to uniformly random permutations of all players in the first round, opening the commitments in the second round. The composition of all the opened permutations determines each player's marginal contribution.  Provided that all
players must send their messages in each round at the same time, the expected reward
of a player is its Shapley value.

However, if the adversary is \emph{rushing} (i.e., in each round, it can
act after seeing an honest player's message),
then in some games, by instructing at most one susceptible player to
abort, it can % from $\mathsf{NaivePerm}$
create a permutation such that
the honest player get its smallest possible reward.
Nevertheless, such an attack is limited to detected violations. This motivates the concept of \emph{violation budget} that
quantifies the number of violations that the adversary can make.
%
 %quantifying the adversary's power based on the number of allowed protocol violations.

%Under the oracle access model for the utility function of the coalitional game, the computational efficiency of a distributed protocol can be measured in terms of the sampling complexity of the utility function.

We also adopt the random permutation approach, and analyze the sampling complexity of random permutations, which we call \emph{P-samples}.
Our idea is that the number of P-samples used in a protocol should depend on two factors: (1) ensuring a small multiplicative error with high probability and (2) mitigating bias from adversaries due to detected violations.

%Based on these security concepts, we will explore the sampling complexity of distributed protocols due to the interplay of these two factors.
%These factors determine the sampling complexity of distributed protocols.

%
%For example, if the approximation protocol is based on sampling permutations in a Monte-Carlo way, one permutation without deviation from the adversary is enough for the expectation of each player to be exactly its Shapley value, while by only this permutation the player is likely to get a reward very different from its expectation. Therefore, we also consider another version of maximin security concerning about the high probability of preserving the rewards.
%An extension of the concept of maximin security concerning about the high probability of preserving the rewards is necessary.

\ignore{
certain criteria in the protocol that are essential for ensuring the security. For the approximation protocols in the Monte-Carlo manner, the number of permutations and the fraction of bad permutations with detected violation determine the execution of the protocol, thus it is natural to establish upper and lower bounds on these factors to achieve the maximin security. Note that the method is inherently based on sampling, making it a form of the sampling complexity, which is challenging to estimate in a perilous environment filled with corrupted players by the adversary.
}

In this paper,
we design maximin secure distributed protocols
that return allocations that approximate the Shapley value.
Moreover, we analyze upper and lower bounds on the sampling
complexity under different violation budget models.
Furthermore, our theoretical results are verified empirically
by performing experiments on both real and synthetic data.

\input{contributions}

%% file: contributions.tex
\subsection{Our Contributions}

\noindent \textbf{Adversarial Model.}
We assess protocol security from an honest player's perspective, assuming other players may be susceptible and controlled by the adversary. The adversary is computationally bounded and cannot break cryptographic primitives. We assume an ideal commitment scheme to simplify presentation. Note that the commitment scheme is the only place we limit the computation power of the adversary; otherwise, the adversary may even perform computation exponential in the number~$n$ of players. The adversary is \emph{rushing} and can wait for an honest player to reveal their commitment before deciding on deviations.

This simplification allows us to pay attention only to deviations from the protocol that can be \emph{detected}.
For example, in Blum's protocol~\cite{tele_coin_flipping}, two players commit to random bits and reveal them in subsequent rounds to determine a \emph{winner}. Deviations in bit sampling by a susceptible player have no effect on the honest player's winning probability, assuming that the adversary cannot break the commitment scheme. On the other hand, if a player fails to open its committed bit, this can be detected from the transcript and we call this a \emph{violation}.

\noindent \textbf{Revisiting $\mathsf{NaivePerm}$.}
To the best of our knowledge, only one previous attempt has been made to design a secure distributed protocol for returning an allocation in a coalitional game. Zlotkin et al.~\cite{sv_protocol} proposed $\mathsf{NaivePerm}$, a random permutation protocol similar to Blum's protocol~\cite{tele_coin_flipping}. In $\mathsf{NaivePerm}$, players commit to uniformly random permutations of all players in the first round, opening the commitments in the second round. The resulting permutation determines each player's marginal contribution.

A rushing adversary can instruct at most one susceptible player to deviate from $\mathsf{NaivePerm}$, creating the permutation such that the honest player get its smallest possible reward. However, the adversary's ability to harm the honest player is limited to detectable deviations or violations. This motivates quantifying the adversary's power based on the number of allowed protocol violations.

\noindent \textbf{Violation Budget.}
The adversary's ability to modify a susceptible player's action is limited by a \emph{violation budget} $C$. This budget represents the maximum number of allowed protocol violations.

%We consider two violation models: perpetual violation, where detected violators are ignored, and transient violation, where each player can violate the protocol only once per round.

\noindent\textbf{Expected vs. High Probability Maximin Security.}
For $0 \leq \epsilon < 1$,
$\epsilon$-expected maximin security of a protocol means
that the expected reward received by an honest player
from the protocol is at least $(1 - \epsilon)$-fraction
of its Shapley value. For $0 < \epsilon, \delta < 1$,
$(\epsilon, \delta)$-maximin security
means that for an honest player,
with probability at least $1 - \delta$,
its received reward is
at least $(1 - \epsilon)$-fraction
of its Shapley value.
Observe that $(\epsilon, \delta)$-maximin security is stronger, as it
implies $(\epsilon + \delta)$-expected maximin security.

\noindent\textbf{P-Samples vs. Round Complexity.}
We focus on distributed protocols using the permutation sampling paradigm outlined in Algorithm~\ref{MainProtocol}. Each permutation sample (P-sample) produces an allocation vector, and the protocol returns the average allocation vector over all P-samples. Since our upper bounds on P-samples can be achieved by the $\mathsf{NaivePerm}$ protocol in $O(1)$ rounds, the round complexity has the same asymptotic upper bounds.

\noindent \emph{Technical Challenges.} Even though the
distributed protocol paradigm with P-samples is simple
and we limit the power of the adversary with a violation budget,
there are still some technical issues we need to resolve.

\begin{compactitem}
\item \textbf{Dependence between P-samples.}
For the
simplest case of known violation budget $C = 1$,
suppose the goal is to achieve $\epsilon$-expected maximin security
using $\mathsf{NaivePerm}$ to generate random permutations.
Then, how many P-samples are sufficient?  If the protocol uses $\Theta(\frac{1}{\epsilon})$ P-samples,
then the adversary can interfere with at most $\epsilon$ fraction
of the P-samples, because the adversary may cause violation at most once.
At first glance, one may conclude that at least $(1 - \epsilon)$-fraction
of P-samples without interference can achieve $\epsilon$-maximin security.
However, since the adversary is rushing,
it may decide to interfere only when it can do the most damage.
Therefore, the P-samples without interference cannot be regarded
as independently identically distributed in the analysis.  As we shall later see,
the correct answer is actually $\Theta(\frac{n}{\epsilon})$ number of P-samples.

\item \textbf{Unknown Violation Budget.}  As aforementioned, the adversary is rushing
and can decide whether to interfere with some P-sample after
it knows what reward is received by an honest player.
Hence, even when there is a large fraction of P-samples without detected
interference from the adversary,
it might still be the case that the resulting reward
for an honest player is small due to the inherent randomness of the sampling process.
Since the adversary may have an arbitrarily large violation budget that is unknown,
a careful argument is needed to analyze the stopping condition of the protocol.

\end{compactitem}

We introduce the parameter $\Gamma$ which is related to the sampling complexity. The detailed definition of $\Gamma$ is provided in Section~\ref{sec:prelim}.

As a warmup, we first give a protocol that achieves
high probability maximin security
when the violation budget~$C$ is known upfront.
Observe that in this case,
actually it is not important whether a violation can be detected.

\begin{theorem}[Known Budget (Informal)]
\label{UB_known_informal}
For any $0 < \epsilon, \delta < 1$,
when the violation budget~$C$ is known,
there exists a protocol that achieves $(\epsilon, \delta)$-maximin security
and uses at most the following number of P-samples:
%$(\epsilon, \delta)$-maximin security can be achieved
%by setting the stopping condition $\mathsf{stopcond}$
%as having the number of permutations reaching:

$$
O\left(\frac{\Gamma}{\epsilon^2} \ln \frac{1}{\delta} +  \frac{C \Gamma}{\epsilon}\right).
$$
\end{theorem}

%For the case where the violation budget is unknown, the stopping condition relies on both the number of permutations and the fraction of violated permutations.

For the case of unknown violation budget,
it is crucial that $\mathsf{NaivePerm}$ can detect violations in order
to achieve the following result,
where the bound is slightly worse than Theorem~\ref{UB_known_informal}.

\begin{theorem}[Unknown Budget (Informal)]
\label{UB_unknown_informal}
For any $0 < \epsilon, \delta < 1$,
when the violation budget~$C$ is unknown (but finite),
there exists a protocol that achieves $(\epsilon, \delta)$-maximin security
such that the protocol always uses
at most the following number of P-samples:

%$(\epsilon, \delta)$-maximin security can be achieved
%by setting the stopping condition $\mathsf{stopcond}$
%as having the number of permutations reaching

%$$\frac{8n}{\epsilon^2} \cdot (\ln \frac{16n}{\epsilon^2} + \ln \frac{1}{\delta})$$
$$O\left(\frac{\Gamma}{\epsilon^2} \ln \frac{\Gamma}{\epsilon \delta} +  \frac{C \Gamma}{\epsilon}\right).$$
%and the fraction of permutations observed so far with detected violation is at most $\frac{\epsilon}{2n}$.
\end{theorem}

\noindent\textbf{Sampling Complexity for Expected Maximin Security.}
In Theorems~\ref{UB_known_informal}
and~\ref{UB_unknown_informal},
the term $O(\frac{\Gamma}{\epsilon^2} \ln \frac{1}{\delta})$
in the number of samples is a result of measure concentration analysis
to achieve high probability statements.

To better understand the other term $O(\frac{C\Gamma}{\epsilon})$,
we investigate the weaker notion of expected maximin security,
and also consider the simpler case where the violation budget~$C$ is known.
The interesting result
is that under this simpler setting,
the number of P-samples is still $\Theta(\frac{C\Gamma}{\epsilon})$
using the random permutation approach.

%In theorem~\ref{UB_known_informal}, the number of necessary permutations is at least $ \max \left\{ \frac{8n}{\epsilon^2} \ln \frac{1}{\delta}, \frac{2 C n}{\epsilon} \right\}$. The term $\frac{8n}{\epsilon^2} \ln \frac{1}{\delta}$ is due to measure concentration analysis to achieve high probability statement. To understand the other term ($\frac{2Cn}{\epsilon}$) better, we investigate a simper case where budget is known and finite, and the aim is to get $\epsilon$-expected maximin security.

%We give the stopping condition for achieving expected maximin security. It suffices to analyze the case with a known budget as (reason ???), therefore the problem is about the sampling complexity of permutations.
%We analyze the number of permutations needed to achieve expected maximin security.
%We first show the upper bound by a relatively stronger ability of violation, transient violation with a known budget.

As a baseline, we analyze the
upper bound of the sampling complexity to achieve expected maximin security using $\mathsf{NaivePerm}$.

\begin{theorem}[Upper Bound (Informal)]
\label{UB_informal_exp}
For any $\epsilon > 0$ with a known violation budget of~$C$,
$R = \frac{\Gamma C}{\epsilon}$ number of P-samples
is sufficient to achieve $\epsilon$-expected maximin security.
\end{theorem}

Since $\mathsf{NaivePerm}$ seems to have very weak security properties,
we try to consider more ``secure'' permutation generation protocols in the context of supermodular games. See section~\ref{sec:prelim} for the definition. In supermodular games, the reward of a player is monotone to its rank in a permutation.
Inspired by Blum's protocol for coin flipping,
we consider $\mathsf{SeqPerm}$,
which puts players sequentially in each position of the permutation,
starting from the least preferable position to the most preferable one.
This permutation generation protocol has better guarantees in the sense that
an honest player will be in one of the top $i$ positions with probability at least $\frac{i}{n}$.
However, even though $\mathsf{SeqPerm}$ seems more secure,
we show that it does not bring a significant improvement
in the sampling complexity. %even when we consider the perpetual violation model.

\begin{theorem}[Lower Bound (Informal)]
\label{LB_informal}
For any sufficiently large enough $n$, there exists a game with $n$ players such that
for any sufficiently small enough constant $\epsilon > 0$,
if the adversary has a violation budget of $1 \leq C \leq \epsilon \Gamma$,
then at least $\Omega(\frac{C\Gamma}{\epsilon})$ number of P-samples generated by $\mathsf{SeqPerm}$
is necessary to achieve $\epsilon$-expected maximin security.
\end{theorem}

Even though the lower bound on the sampling complexity
in Theorem~\ref{LB_informal} is stated for $\mathsf{SeqPerm}$,
in Remark~\ref{remark:LB}, we
will give the general conditions for a permutation generation protocol
under which this lower bound holds.

%% file: related.tex
%\vspace{10pt}

%\noindent \textbf{Related Work.}
%\subparagraph*{Related Work.}
\subsection{Related Work}
%Among all ways of distributing rewards to players,
Shapley value~\cite{sv_origin} satisfies symmetry, group rationality and additivity, and has become the \emph{de facto} fairness notion.
%For supermodular games, the Shapley value is in the \emph{core}~\cite{sv_core}.
Shapley value also
has applications areas such as evaluating the importance of data sources in a learning model~\cite{sv_by_diff_matrix} and assessing the contribution of data to database queries~\cite{query1,query2,query3}.

However, computing the Shapley value is intractable. Deng et al.~\cite{sv_hardness} proved that calculating the Shapley value for weighted voting games is \#P-complete.
%Therefore, finding an efficient algorithm for the general case is unlikely.
Elkind et al.~\cite{wvg_hardness} proved that it is NP-hard to decide whether the Shapley value of a certain player is 0 in the weighted voting game;
hence, it is NP-hard to achieve constant multiplicative error for general games.
Bachrach et al.~\cite{ss_index_hardness} proved the impossibility of achieving super-polynomially small additive error with high probability
if the randomized algorithm only samples the utility function
for a polynomial number of times.

Supermodular coalitional game models the situation in which
it is more beneficial for players to collaborate together,
an example of which is multicast cost sharing~\cite{convex_game_1}. Games not having supermodular utility functions
may still be reduced to supermodular games; see~\cite{convex_game_2, convex_game_3, sv_hardness} for instances.
However, computing the Shapley value exactly is still hard for supermodular games.
In fact, Liben-Nowell et al.~\cite[Theorem 6]{sv_convex}
showed that any randomized approach will still need
at least $\Omega(\frac{1}{n \epsilon})$ oracle accesses to the utility function
to achieve $\epsilon$-multiplicative error.
Similar to the result in~\cite{ss_index_hardness}, this means that a super-polynomially small $\epsilon$ would require
a super-polynomially number of accesses to the utility function.
Moreover, they also showed~\cite[Theorem 5]{sv_convex} that any deterministic
algorithm that achieves a multiplicative error of at most $\epsilon = \frac{1}{2n}$
must make at least a super-polynomial number of samples on the utility function.
In contrast, for this value of $\epsilon$,
a polynomial number of samples is sufficient for randomized algorithms.

%To the best of our knowledge, all approximation algorithms
%on the Shapley value in the literature are randomized.
%As inspired by the definition of Shapley value,
Approximation approaches by sampling random permutations
have been well studied,
which is known as \emph{simple random sampling}~\cite{sv_by_srs, sv_protocol}.
More refined analysis with respect to the ranks in random permutations,
which is known as \emph{stratified sampling}~\cite{sv_stratify_range, sv_stratify_var},
has also been studied.  Finally,
by re-interpreting the Shapley value as the expectation
of a process that randomly samples subsets of players,
multilinear sampling methods~\cite{sv_by_multilinear}
have also been investigated. Other heuristic algorithms can be found in~\cite{sv_by_24methods}.

The sampling complexity of algorithms has also been extensively studied.
For general games, a simple application of the Hoeffding inequality can
give a sufficient number of permutation samples.  However, note that each permutation sample involves $\Theta(n)$ oracle accesses to the utility function,
but only two of which are relevant to a single player.
To avoid this extra factor of $\Omega(n)$ in the sampling complexity
on the utility function, Jia et al.~\cite{sv_by_diff_matrix}
has considered a different sampling procedure that first approximates
the differences of Shapley values between all pairs of players.

For the special case of supermodular games,
Liben-Nowell et al.~\cite{sv_convex} used the Chebyshev's inequality
to give an upper bound of the permutation sampling complexity
to achieve multiplicative approximation.  We show that
this can be easily improved by using the stronger Chernoff Bound.

%There is not too much prior work that investigates
%distributed protocols that approximate the Shapley value
%in an adversarial setting.
Zlotkin et al.~\cite{sv_protocol} proposed
a distributed protocol (which we call $\mathsf{NaivePerm}$)
that samples one random permutation,
from which an allocation is derived.
As shown in Section~\ref{algorithm},
just using one sample is not sufficient under our adversarial model.

The game-theoretic notion of \emph{maximin security}
has been first proposed by Chung et al.~\cite{maximin}
to analyze how the reward for an honest player in a distributed protocol
can be protected in an environment where other players
may behave adversarially.
Subsequently, this notion has been applied to distributed protocols
for other well-known problems such as
\emph{leader election}~\cite{leader_election} and
\emph{ordinal random assignment problem}~\cite{DBLP:conf/acns/ChanWXX23}.
Both upper and lower bounds on the round complexity
have been analyzed in~\cite{leader_election}.

\ignore{
In contrast, our upper bounds on the round complexity
is derived from the upper bounds on P-samples,
and our lower bound is only for P-samples under
certain conditions on the permutation generation protocols.
}

\ignore{

\subsection{Paper Organization}

In Section~\ref{sec:prelim}, we define the formal
notation and clarify our settings.
In Section~\ref{algorithm}, we
describe the general distributed protocol using the
random permutation paradigm,
and revisit the permutation generation protocol $\mathsf{NaivePerm}$.
In Section~\ref{sec:high_prob},
we give upper bounds for the sampling complexity
of protocols that achieve high probability maximin security,
where the highlight is a protocol that does not need to
know the violation budget upfront.
In Section~\ref{sec:lowerbound},
we analyze both the upper and the lower bounds on the sampling complexity
to achieve expected maximin security.
Finally, in Section~\ref{sec:conclusion},
we conclude by suggesting some future research directions.
} 

%% file: preliminary.tex
\section{Preliminaries}
\label{sec:prelim}

\subsection{Coalitional Game}

A \emph{coalitional game} is characterized by $(\all, v)$,
where $\all$ is a set of $n = |\all|$ players and $v : 2^{\mathcal N} \to \mathbb R_+$ is a utility function indicating that a coalition $S$ of players can obtain a revenue of $v(S)$.

\noindent \textbf{Marginal Contribution.}
The \emph{marginal contribution} of player $i$
joining a subset $S \not\ni i$
is:
$\mu_i(S) := v(S \cup \{i\})-v(S)$. Observe that $\mu_i$ is defined with respect to $v$.
In this work, %apart from being non-negative,
we assume that the utility function~$v$ is monotone,
i.e., the marginal contribution $\mu_i$ is always non-negative.

\ignore{
\begin{assumption}[Non-negative Monotone Game]
	\label{assume}
	We assume the function $v: 2^\all \rightarrow \R$ satisfies:
	\begin{itemize}
		\item \emph{Non-negative.} For $\forall S \subseteq \all$, $v(S) \ge 0$.
		\item \emph{Monotone.} For $\forall S \subseteq T \subseteq \all$, $v(S) \leq v(T)$. \quan{No need ``for"? }
	\end{itemize}
\end{assumption}
}
\ignore{
In this work, we use monotonicity in Assumption~\ref{assume}
via the following fact. \quan{Is this fact useful? Maybe delete this, but mention $\mu_i(S) \leq \mu_i(T)$ if $S \subseteq T$ in convex games.}
\begin{fact}[Non-negativity of Marginal Contribution]
\label{MMC}
If a game $v$ is monotone, then for every player $i$ and any subsets $S  \subseteq \all$ such that $i \notin S$,
\begin{equation}
  \mu_i(S) \geq 0.
\end{equation}
\end{fact}
}

\noindent \textbf{Dividing Up the Reward in the Grand Coalition.}
From monotonicity, we have for all $S \subseteq \all$, $v(S) \leq v(\all)$. An interesting question is how to divide up the received reward $v(\all)$ among
the players.
%\quan{It would be good if we can give reference saying that finding a stable coalition is NP-hard for monotone games. }
Given a game~$(\all, v)$, an \emph{allocation} is a (non-negative) vector $x \in \R_+^\all$
such that $\sum_{i \in \all} x_i = v(\all)$,
i.e., $x_i$ represents
the profit allocated to player~$i$.
% We use $\mathcal A$ to denote the set of allocations (with respect to~$v$).

\noindent \textbf{Shapley Value~\cite{sv_origin}.}
This is an allocation vector whose definition is recalled as follows.
For player~$i$, its Shapley value is:
\begin{equation}\label{sv}
  \begin{aligned}
  \phi_i &= \frac{1}{n} \sum_{S \subseteq \all \setminus \{i\}} \binom{n-1}{|S|}^{-1} \mu_i(S)
  = \frac{1}{n!} \sum_{\sigma} \mu_i(P_i(\sigma)),
  \end{aligned}
\end{equation}
where the last summation is over all $n!$ permutations $\sigma: [n] \to \all$ of players, and $P_i(\sigma)$ is the set of \emph{predecessors} of player~$i$ in the permutation $\sigma$.

\noindent \textbf{Random Permutation Interpretation.}
One way to interpret Equation~(\ref{sv}) is to distribute the total reward $v(\all)$ through a random process. This process involves sampling a uniformly random permutation (\emph{aka} P-sample) that determines the order in which players join the coalition one by one. The reward of a player is its marginal contribution when joining existing players already in the coalition. The Shapley value $\phi_i$ is exactly the expected reward of player~$i$ in this random process.  The following parameter will be used to analyze
the number of P-samples used in our protocols.

%We introduce the \textit{Max-to-Mean} Parameter $\Gamma$ of a game, which will be utilized in future analyses.

\noindent \textbf{Max-to-Mean Ratio~$\Gamma$.}
This parameter is determined by the utility function~$v$.
For player~$i$, let $U_{\max, i} := \max_{S \subseteq \all \backslash \{i\}} \mu_i(S)$ be its maximum marginal contribution;
its \textit{max-to-mean ratio} is
$\Gamma_i = \frac{U_{\max, i}}{\phi_i}$,
where we use the convention $\frac{0}{0} = 1$.
%When the context is clear, we will also omit the subscript~$i$.
We denote $\Gamma := \max_{i \in \all} \Gamma_i$.

\noindent \textbf{Expectation vs High Probability.}
%Even though the expected reward in the above random process is exactly the Shapley value, using only one random permutation has a high variance.
%A standard way to resolve this is to repeat the process using many independent random permutations and take the average reward.
Typical measure concentration results, such as Chernoff Bound (Fact~\ref{fact:chernoff}), show that the number of samples needed to accurately estimate the mean increases linearly with the max-to-mean ratio~$\Gamma$.

\noindent \textbf{Source of Randomness.}
The above sampling methods assume that there is an unbiased source of randomness.
In this work, we are interested in the scenario where there is no central authority to generate the randomness.
Instead, randomness needs to be jointly generated from the players via distributed protocols.

\subsection{Distributed and Security Model}
\label{SecureSetting}

We clarify the model of distributed protocols that we will employ.

\noindent \textbf{Communication Model.} We consider a \emph{synchronized} communication model, in which each player can post messages to some broadcast channel (such as a ledger \cite{DBLP:conf/crypto/BadertscherMTZ17}) in \emph{rounds}. We assume that the \emph{Byzantine broadcast} problem has already been solved, i.e., a message posted by any player (even adversarial) in a round will be seen by all players at the end of that round.

\noindent \textbf{Distributed Protocol Model.}
A distributed protocol specifies the behavior of each player in $\all$, each of which has its own independent source of randomness.
\ignore{
In each round, each player reads posted messages on the channel from previous rounds, performs some local computation (possibly based on locally generated randomness) and posts new messages to the channel. The protocol \emph{terminates} when all players should stop sending messages. Typically, it is desirable that a protocol terminates in a polynomial number of rounds.

A distributed protocol also implicitly specifies an \emph{evaluator} which is some publicly agreed deterministic function that is applied to the whole transcript of messages after termination, and returns the \emph{output} of the protocol. However, for ease of exposition, the roles of a player and the evaluator are not explicitly distinguished in our protocol description.
In general, the computation that is necessary to generate a message is performed by a player, while other deterministic computation that can be carried out using the transcript is implicitly performed by the evaluator.

Observe that if some players violate the protocol, the protocol might not terminate, or the evaluator might return some default value~$\bot$.

In the description of a protocol, the roles of a player and the evaluator may not be explicitly distinguished.

For ease of exposition, we also describe a protocol as if the evaluator's computation is carried out during the protocol execution, but we note that these computations can actually be deferred when the protocol terminates.
}
In the context of a coalitional game, we assume %that the players and the evaluator have
oracle access to the utility function~$v$.

\ignore{
\noindent \emph{Complexity of the evaluator.}
It is desirable that the evaluator only needs
to perform polynomial-time computation.
As a trivial example, an uninteresting protocol may specify that the players
do not need to send any messages, and the evaluator just
uses an exponential number of oracle accesses to $v$ and computes the Shapley values itself.
}

\noindent \textbf{Adversarial Model.}
We analyze the protocol from the perspective of some honest player~$i^*$,
%who is \emph{immune} to a \emph{Byzantine adversary}. From the perspective of an honest player~$i^*$,
where all other players in $\all \setminus \{i^*\}$ may be \emph{susceptible} to the adversary. Specifically, for a susceptible player, the adversary can control its randomness, observe its internal state, and dictate its actions.
We assume that the adversary is \emph{rushing}, i.e., it can wait and see the messages from player~$i^*$ in a round before it decides the actions of the susceptible players in that round.

\ignore{
We also assume that the adversary is computationally bounded in the sense that it cannot break cryptographic primitives such as commitment schemes. As we shall see, this allows us to design protocols such that the only way an adversary can harm an honest player is to instruct susceptible players not to broadcast certain messages. This means that such a \emph{violation} from the protocol can be \emph{detected}. We use the term "violation" to refer to a deviation that can be detected by the evaluator from the transcript. On the other hand, deviating from the protocol's description in sampling randomness cannot be detected in our model and is therefore not considered a violation.
}

\noindent \textbf{Commitment Scheme.}
Assuming the existence of one-way functions/permutations, there is a constant-round publicly verifiable commitment scheme \cite{DBLP:journals/jacm/LinP15} that is perfectly correct, perfectly binding, and concurrent non-malleable.
The \emph{commit} operation allows a player to construct a \emph{commitment} $c$ of some secret message $m$. The \emph{open} operation allows the player to reveal the secret message~$m$ from previously committed~$c$, where perfectly binding means that it is impossible to open the commitment to any other message different from $m$. For simplicity, we assume that each of the commitment operation and the open operation can be performed in a single round.

\noindent \emph{Ideal Commitment Scheme.}
For ease of presentation,
we assume that the commitment scheme is also
\emph{perfectly hiding}.
This means that the adversary cannot learn anything about the secret message from
the commitment~$c$. Other than this restriction,
we allow the adversary to perform any other computation (even exponential in~$n$).

In a real-world scheme, the commitment is only \emph{computationally hiding}.
Formally, for some security parameter~$\lambda$,
this would introduce an extra multiplicative factor
$(1 - \negl(\lambda))$ in the reward of the honest player for some negligible function $\negl(\cdot)$.
However,  since we will be considering $(1 - \epsilon)$-approximation for constant $\epsilon > 0$,
this extra factor of $(1 - \negl(\lambda))$ may be absorbed into $\epsilon$.

\noindent \textbf{Protocol Violation.}
We shall see that because of the ideal commitment scheme,
the only way an adversary can harm an honest player is
to instruct a susceptible player to refuse opening some
previously committed message, where such a \emph{violation}
can be detected.  On the other hand, deviating from the protocol's description in sampling randomness cannot be detected in our model and is therefore not considered a violation.  Below are violation models that we consider.

\begin{itemize}
\item \emph{Violation Budget.} The violation budget~$C$
is the maximum number of violations that the adversary can make
throughout the protocol, where $C$ can either be known or unknown to the protocol.

%In fact,
%we may even assume that one unit of the budget
%can make any number of corrupted users deviate in a single round.

\item \emph{Violation Rate.}
The violation rate~$f \in [0,1]$ means that
for every $T > 0$, during the generation of the first $T$ number of P-samples,
there can be at most $f T$ violations.
\end{itemize}

We will investigate the following notions of security for our protocols.

\begin{definition}[Maximin Security for Shapley Value]
\label{defn:maximin}
Suppose $\Pi$ is a (randomized) distributed protocol between players~$\all$
in some coalitional game $(\all, v)$, where there is at least one honest user.
The purpose of $\Pi$ is to return an allocation  that is close to the Shapley vector~$\phi$.
Suppose that when $\Pi$ is run against some adversary~$\Adv$, it always terminates
with some output allocation vector~$x \in \R_+^\all$, i.e., $\sum_{i \in \all} x_i = v(\all)$.

\begin{itemize}

\item \emph{Expectation.} We say that $\Pi$
is $\epsilon$-expected maximin secure against~$\Adv$
if, for any honest player~$i^*$,
its expected allocation satisfies $\E[x_{i^*}] \geq (1 - \epsilon) \cdot \phi_{i^*}$.

\item \emph{High probability.}
We say that $\Pi$ is $(\epsilon, \delta)$-maximin secure
against~$\Adv$,
if, for any honest player~$i^*$,
with probability at least $1 - \delta$,
its allocation is $x_{i^*} \geq (1 - \epsilon) \cdot \phi_{i^*}$.

\emph{Adaptive Security.}  In this case,
the protocol $\Pi$ returns an allocation $x$ and also a parameter~$\epsilon \in [0,1]$.
We say that it is $\delta$-adaptively maximin secure,
if, with probability at least $1 - \delta$,
$x_{i^*} \geq (1 - \epsilon) \cdot \phi_{i^*}$
\end{itemize}
\end{definition}

\subsection{Max-to-Mean Ratio~$\Gamma$ in Special Games}
\label{sec:gamma}

%The ratio~$\Gamma$ can be large for a general monotone game.

\begin{lemma}[$\Gamma$ Ratio for Monotone Games]
\label{lemma:gamma_monotone}
In a game where the marginal contribution
of every player is non-negative,
$\Gamma \leq n \cdot {n-1 \choose \lfloor \frac{n-1}{2} \rfloor}$.
\end{lemma}

\begin{proof}
For any player $i$, $\Gamma_i$ is maximized if there's exactly 1 set $Q \subseteq \all \backslash \{i\}$ such that $\mu_i(Q) > 0$. Let $k = |Q|$. In this case, $\Gamma_i = n \cdot {n-1 \choose k}$, which is maximized when $k = \lfloor \frac{n-1}{2} \rfloor$. Therefore, we can conclude that $\Gamma \leq n \cdot {n-1 \choose \lfloor \frac{n-1}{2} \rfloor}$ for general monotone games.
\end{proof}

\noindent \textbf{Game with Maximum $\Gamma$. }Next, we will define a game with $|\all| = n$ players where $\Gamma$ is maximized.
Let $S_0 \subseteq \all$ be a subset of players such that $|S_0| = \lfloor\frac{n-1}{2} \rfloor $. Consider the following game $v: 2^{\all} \rightarrow \R_+$.

$$
v(S) =
\begin{cases}
	1, & S_0 \subsetneq S, \\
	0, & \text{Otherwise}
\end{cases}
$$

We can categorize the players based on whether they belong to the set $S_0$, where players in the same category are considered symmetric.

In this game, $\Gamma_i$ is maximized when $i \notin S_0$. Because for any player $i \notin S_0$, we have $\mu_i(S) = 1$ if and only if $S_0 = S$. Otherwise, $\mu_i(S) = 0$. Thus, $U_{\max, i} = 1$. The Shapley value $\phi_i$ can be simplified to:

$$
\phi_i = \frac{1}{n}\cdot \sum_{S \subseteq \all \backslash \{i\}} {n-1 \choose |S|}^{-1} \mu_i(S) = \frac{1}{n} {n-1 \choose |S_0|}^{-1}
$$

Since $|S_0| = \lfloor\frac{n-1}{2} \rfloor$,  $\Gamma_i = \frac{1}{{n-1 \choose |S_0|}^{-1}/n} = n \cdot {n-1 \choose \lfloor \frac{n-1}{2} \rfloor}$. This expression represents the maximum value of $\Gamma$ among all monotone games.

\subsection{Properties of Supermodular Games}
\label{sec:supermodular}

In additional to non-negativitiy and monotonicity of~$v$,
a supermodular (or convex) game
is associated with a supermodular utility function
$v: 2^\all \rightarrow \R_+$.
In other words, for all $ S, T \subseteq \all$, $v(S) + v(T) \le v(S \cup T) + v(S \cap T)$.

Supermodularity implies the following fact.
\begin{fact}[Monotonicity of Marginal Contribution]
	\label{MMC}
	If a game $v$ is supermodular, then for every player $i$ and any subsets $S \subseteq T \subseteq \all$ such that $i \notin T$,
	\begin{equation}
		\mu_i(T) \ge \mu_i(S).
	\end{equation}
\end{fact}

It is known that Fact~\ref{MMC}
implies that
the most profitable way for all players is to cooperate all together in the grand coalition~$\all$.
Formally, this means that there does not exist
$k \geq 2$ disjoint subsets $S_1, \ldots, S_k$
such that $\sum_{j=1}^k v(S_j) > v(\all)$.
Therefore, the Shapley value is particularly applicable to supermodular games, as the grand coalition is most preferable to every player and the Shapley value provides a meaningful solution to the grand coalition.

\noindent \textit{Rank Preference.} For $j \in [n]$, define
$U_j := \E[\mu_{i^*}(P_{i^*}(\sigma)) | \sigma(i^*) = j]$,
which is the expected reward of player~$i^*$ conditioning on it
having rank~$j$ in $\sigma$. The following lemma justifies why
rank 1 is the least preferable in supermodular games.

\begin{lemma}[Monotonicity of Ranks]
	\label{lemma:ranks}
	For $1 \leq j < n$, $U_{j} \leq U_{j+1}$.
\end{lemma}

\begin{proof}
	We use Fact~\ref{MMC} in the first inequality:
\begin{align*}
	U_{j+1} & =  \frac{1}{{n-1 \choose j}} \sum_{S \in {\all\setminus\{i^*\} \choose j}} \mu_{i^*}(S)
	\\ & \geq \frac{1}{{n-1 \choose j}} \sum_{S \in {\all\setminus\{i^*\} \choose j}} \frac{1}{j} \sum_{a \in S} \mu_{i^*}(S \setminus \{a\}) &&
	\textrm{\small ($j$ ways to remove $a \in S$)} \\
	& =  \frac{1}{{n-1 \choose j}} \cdot \frac{n-j}{j} \sum_{T \in {\all\setminus\{i^*\} \choose j-1}} \mu_{i^*}(T)  &&
	\textrm{\small ($n-j$ possible $S$'s to form $T$)}\\
	& = \frac{1}{{n-1 \choose j-1}} \sum_{T \in {\all\setminus\{i^*\} \choose j-1}} \mu_{i^*}(T) = U_j &&
\end{align*}
\end{proof}

\begin{fact}[$\Gamma$ Ratio for Supermodular Games]
	\label{fact:max_ratio}
	For supermodular games, $\Gamma \leq n$.
	%$U_{\max} = U_n \leq n \E[U] = n \phi_{i^*}.$
\end{fact}

\begin{proof}
	Since player~$i^*$ is in the most preferable rank~$n$ with
	probability $\frac{1}{n}$,
	we have $\phi_{i^*} = \E[U] \geq \frac{1}{n} \cdot U_{\max,i^*}$,
	where the inequality holds because $U \geq 0$.
\end{proof}

\subsection{Properties of Edge Synergy Game}
\label{sec:synergy}

The edge synergy game \cite{sv_hardness} was defined on a simple undirected graph, but we extend it to a weighted hypergraph $G = (\all, E, w)$, where each vertex in $\all$ is a player, $E \subseteq 2^\all$ and $w: E \rightarrow \R_+$ gives the hyperedge weights.
The utility function $v$ is defined as $v(S) := \sum_{e \in E: e \subseteq S} w(e)$, which is the sum of edge weights in the induced subgraph $G[S]$.  See more details in \Cref{sec:synergy}.

\begin{fact}
The edge synergy game is supermodular and has max-to-mean ratio $\Gamma \leq \max_{e \in E} |e|$.
\end{fact}

\begin{proof}
\noindent \textbf{Supermodularity of Synergy Game.}
This game is supermodular because for any subsets $S, T \subseteq V(G)$, $E(G[S]) \cup E(G[T]) \backslash E(G[S \cap T]) \subseteq E(G[S \cup T])$.

\noindent \textbf{$\Gamma$ ratio of Synergy Game.}
Let $i$ be a vertex, and $d(i)$ denote its weighted degree.
By Fact~\ref{MMC}, $U_{\max, i} = \mu_i(\all \backslash \{i\}) = \sum_{e \in E: i \in e} w_e = d(i)$.
The Shapley value of $i$ can be computed as follows.

\begin{enumerate}
\item For normal graphs, we have:
$$
	\begin{aligned}
		\phi_i = \frac{1}{n!} \sum_{\sigma} \mu_i(P_i(\sigma)) = \sum_{\sigma} \sum_{\substack{ (i,j) \in E(G),\\ j \in P_i(\sigma)}} \frac{w_{i,j}}{n!} \\= \sum_{(i, j) \in E(G)}\sum_{\sigma : j \in P_i(\sigma)} \frac{w_{i,j}}{n!} = \sum_{(i, j) \in E(G)}  \frac{w_{i,j}}{2} = \frac{d(i)}{2}
	\end{aligned}
$$

Therefore, for every vertex $i$, we have $\Gamma_i = 2$. Thus, we can conclude that $\Gamma = 2$ for the edge synergy game on simple graphs.

\item For hypergraphs, we have:
$$
\begin{aligned}
	\phi_i = \frac{1}{n!} \sum_{\sigma} \mu_i(P_i(\sigma)) = \sum_{\sigma} \sum_{\substack{e \in E(H), \\i\in e, \\e\backslash \{i\} \subseteq P_i(\sigma)}} \frac{w_e}{n!} \\= \sum_{\substack{e \in E(H),\\ i \in e }}\sum_{\sigma: e \backslash \{i\} \subseteq P_i(\sigma)} \frac{w_e}{n!} = \sum_{\substack{ e \in E(H), i\in e}}  \frac{w_e}{|e|}
\end{aligned}
$$

Therefore, for every vertex $i$, $\Gamma_i = \frac{d(i)}{\sum_{\substack{ e \in E(H), i\in e}} \frac{w_e}{|e|}} \leq \max_{e \in E} |e|$.
\end{enumerate}
\end{proof}

%% file: algorithm.tex
\section{Distributed Protocols Based on Permutation Sampling}
\label{algorithm}

\noindent \textbf{High Level Approach.}
Based on the random permutation
approach in equation~(\ref{sv}),
we give an abstract protocol in Algorithm~\ref{MainProtocol}.
It uses a sub-protocol $\mathsf{GenPerm}$
that
%takes a subset of players~$N \subseteq \all$
%and
returns a random permutation $\sigma$ of $\all$,
which we shall call a \textbf{P-sample}.
\ignore{
The condition $N \neq \all$ arises exclusively during the analysis of the lower bound in supermodular games in section~\ref{sec:lowerbound}. In such cases, the expected reward of a player is monotonic with respect to its rank in $\sigma$. Players not in $N$ are assumed to occupy the lowest ranks (representing the least favorable positions) in $\sigma$.
}
Given a P-sample $\sigma$,
we indicate the player at rank~$j$ by $\sigma_j$,
and $\sigma_{<j}$ represents the set of players ranked lower than $\sigma_j$. In other words, $\sigma_{<j}$ can be defined as $\{ \sigma_i | i < j\}$. %\quan{Should we mention here that the detected violated players should be removed. If so, we can consider move perpetual/transient here. }

The sub-protocol $\mathsf{GenPerm}$ also returns
other auxiliary information $\mathsf{aux}$ such as
whether any player has been detected for violation.

\noindent \emph{Stopping Condition.}
To decrease the variance of the output and
mitigate the effect of the adversary,
we may use multiple numbers of P-samples
and take the average allocation over all samples.
The stopping condition $\mathsf{stopcond}$ specifies
when the whole protocol should terminate.
\ignore{It may be a rule
that states upfront the required number of P-samples,
or it may depend on the behavior of violating players
during the execution of the protocol.
}

\begin{algorithm}
\caption{Abstract Protocol to Return an Allocation}\label{MainProtocol}
\SetKwFunction{GenPerm}{GenPerm}
\KwIn{A game $(\all, v)$, a stopping condition $\mathsf{stopcond}$.}
\KwOut{An allocation $x \in \R_+^\all$.}
%$N \gets \all$ \\
$z \gets \vec{0} \in \R_+^\all$,  $R \gets 0$ \\
\While{$\mathsf{stopcond} = \mathsf{false}$}{
    $(\sigma, \aux) \gets$ \GenPerm{$\all$} \\

    %\Comment{$\sigma_{\all \backslash N}$ is an arbitrary permutation of players in $\all \backslash N$. }
    %$\sigma \leftarrow \sigma  \bigoplus \sigma_{\all \backslash N}$

		\Comment{Marginal contribution of rank~$j$ player in $\sigma$}
		\textbf{for every} $j \in [n]$,
    $z_{\sigma_j} \gets z_{\sigma_j} + \mu_{\sigma_j}(\sigma_{<j})$ \\
		$R \gets R+1$ \\
		\Comment{May update other variables according to $\aux$}
		
}
\Return $x := \frac{1}{R} \cdot z \in \R_+^\all$
\end{algorithm}

\begin{remark}[Reward Decomposition]
\label{remark:decompose}
In our analysis of Algorithm~\ref{MainProtocol}, it will be convenient
to decompose the reward received
by the honest player~$i^*$ in each P-sample as follows.

Specifically, in the $j$-th P-sample,
the reward received
can be expressed as $X_j := Y_j - Z_j$,
where $Y_j$ is the reward received by player~$i^*$
had the adversary not caused any violation (if any),
and $Z_j$ can be interpreted as the amount of damage
due to the adversary's violation in the $j$-th P-sample.

This allows us to analyze $Y_j$ and $Z_j$ separately.
For instance, because there is at least one
honest player, we can assume that the $Y_j$'s are independent among themselves,
and $\E[Y_j] = \phi_{i^*}$.

On the other hand, the random variables $Z_j$'s can interact
with the $Y_j$'s and behave in a very complicated way.
However, if we know that the adversary has a violation budget of $C$,
then we can conclude that with probability 1,
$\sum_j Z_j \leq C U_{\max, i^*} = C\Gamma \cdot \phi_{i^*}$.
\end{remark}

\ignore{
\begin{remark}
Since in each P-sample the sum of the marginal contribution of each player is $v(\all)$, the vector $x$ we get at the end of the protocol satisfies $\sum_{i \in \all} x_i = v(\all)$ as well, which means $x$ is an allocation.
\end{remark}

\begin{remark}
Observe that each player just needs to participate in $\mathsf{GenPerm}$
and also performs the computation
to decide if the stopping condition is satisfied.
The deterministic computation of averaging the allocation vectors
is performed by the evaluator,
which has oracle access to the utility function~$v$ (and hence, can evaluate
the marginal contribution $\mu$).
\end{remark}

Different variants of the abstract protocol
depend on how we implement $\mathsf{GenPerm}$
and what stopping condition $\mathsf{stopcond}$
is used.
}

\subsection{``Secure'' Permutation Generation}

We revisit a simple permutation generation protocol
that was proposed by~\cite{sv_protocol}.
We paraphrase it in Algorithm~\ref{NaivePerm},
and call it $\mathsf{NaivePerm}$, because
it can easily be attacked by a rushing adversary in supermodular games(see Lemma~\ref{lemma:naive_attack}).
However, as we shall later see in \Cref{sec:high_prob,sec:lowerbound},
the security property in Lemma~\ref{lemma:naive_secure}
turns out to be sufficient for us to design good protocols.

%
%its security implicitly assumes that the adversary
%is non-rushing, which is a very strong assumption
%because a fair coin can be simulated between
%2 players in a constant number of rounds under this assumption
%(contrasting the Cleve's impossibility result~\cite{???????}).

\begin{algorithm}
\caption{$\mathsf{NaivePerm}$}\label{NaivePerm}
\SetKw{KwDownTo}{downto}
\SetKwFunction{RandElim}{RandElim}
\KwIn{Player set $N$.}
\KwOut{A random permutation $\sigma$
of $N$, together with a collection $\mathsf{Dev}$
of detected violating players.}

%Assume that the players in $\all$
%can be linearly ordered (e.g., using the lexicographical order),
%which is used to determine the composition order later.

%Linearly order players in $S$, call them $S_0, \cdots, S_{|S|-1}$. (e.g. according to the lexicographical order.) \\

\Comment{Commit phase:}
\ForEach{{\rm \bf player} $i \in N$} {
    Uniformly sample a permutation $\sigma^{(i)}$ of $N$. \\
    Commit $\sigma^{(i)}$ and broadcast the commitment. \\
}

\Comment{Open phase:}
\ForEach{{\rm \bf player} $i \in N$} {
    Broadcast the opening for its previously committed $\sigma^{(i)}$. \\
}

Denote $S$ as the collection of
players~$i$ that have followed the commit and the open phases,
i.e., its commitment and opened $\sigma^{(i)}$
have been verified.  \\

Using the permutations generated by
players in $S$, compute the composed permutation $\sigma \gets \circ_{i \in S} \, \sigma^{(i)}$,
using the predetermined composition order (such as
the lexicographical order of players).  \\

$\mathsf{Dev} \gets N \setminus S$ \\

\Return $(\sigma, \mathsf{Dev})$
\end{algorithm}

\begin{lemma}[Security of $\mathsf{NaivePerm}$]
\label{lemma:naive_secure}
Suppose in Algorithm~\ref{NaivePerm}, the strategy of an adversary $\Adv$ never causes
violation.
Then, assuming an ideal commitment scheme,
Algorithm~\ref{NaivePerm} returns a uniformly random permutation of players
if there exists at least one honest player.
\end{lemma}

\begin{proof}
Because of the ideal commitment scheme assumption,
all players must follow the commit and the open phases to avoid violation being detected.
Since there exists an honest player~$i$, its sampled permutation $\sigma^{(i)}$ is uniformly random,
whose commitment is assumed to leak no information.
Therefore, other players' sampled permutations are independent of~$\sigma^{(i)}$,
and so, the resulting composed permutation is uniformly random.
\end{proof}

\subsection{Attack $\mathsf{NaivePerm}$ in Supermodular Games}

Recall that by Lemma~\ref{lemma:ranks}, in supermodular games, the expected reward of a player is monotone to its rank in $\sigma$. The following Lemma shows $\mathsf{NaivePerm}$ can be attacked by a rushing adversary so that the honest player is always at the least favorable position.

\begin{lemma}[Attacking $\mathsf{NaivePerm}$ Using at Most One Abort in Supermodular Games]
	\label{lemma:naive_attack}
	%Suppose $n$ is odd.
	Suppose the game is supermodular and a rushing adversary controls $n-1$ players.
	Then, by instructing at most 1 player to abort,
	it can always cause $\mathsf{NaivePerm}$ to return
	a permutation such that the remaining honest player~$i^*$
	is at the least preferable position with probability 1.
\end{lemma}

\begin{proof}
	For simplicity, consider the case that the players are indexed by $[n]$,
	where $i^*=n$ is the only honest player.
	The adversary adopts the following strategy.  Pick any cycle permutation~$\tau$, i.e.,
	the permutation corresponds to a cyclic shift of the players~$\all$.  Note that by
	applying the shift $n$ times, we have $\tau^n$ equal to the identity permutation.
	
	Then, the $j$-th susceptible player will commit to $\tau^j$.
	The key observation is that
	for $0 \leq i \leq n-1$,
	the $n$ numbers $\sum_{j=1}^{n-1} j - i$
	are distinct modulo $n$.
	This means that by omitting the permutation from at most 1 susceptible player,
	the adversary can simulate a cyclic shift for
	$i$ positions, for any $0 \leq i \leq n-1$.
	
	%Observe that
	%for odd $n$, $\sum_{j=1}^{n-1} j$ is divisible by $n$.  Therefore, by omitting the
	%permutation $\tau^j$ from the $j$-th corrupted player,
	%the composition of permutations from the remaining corrupted players
	%correspond to $\tau^{-j}$, which is a cyclic shift of $j$ positions.
	
	Therefore, for a rushing adversary,
	after the honest player~$i^*$ opens its committed permutation $\sigma^{(i^*)}$,
	the adversary can observe the rank of $i^*$ in its permutation.
	Hence, the adversary can always instruct at most 1 susceptible player to abort opening its
	committed permutation
	such that the honest player is shifted to the least preferable position.
\end{proof}

\begin{remark}
	In Lemma~\ref{lemma:naive_attack}, even if the single violating player
	is punished by sending it to the least preferable position instead of $i^*$,
	player~$i^*$ is still at the second least preferable position.
\end{remark}

\noindent \textbf{More ``Secure'' Permutation Generation.}
Although $\mathsf{NaivePerm}$ has very weak security properties especially in supermodular games, as discussed in \Cref{sec:high_prob}, it is still useful because it can detect violations. This feature allows us to achieve high probability maximin security even when the violation budget is unknown. In \Cref{sec:lowerbound}, we will investigate a more secure P-sample protocol and analyze its sampling complexity to achieve expected maximin secure for supermodular games. However, we will find that, surprisingly, the stronger security properties do not lead to a significant advantage in terms of the sampling complexity. %\quan{(and analyze the expected maximin secure for supermodular games)}

%% file: stability.tex
\section{Sampling Complexity to Achieve High Probability Maximin Security}
\label{sec:high_prob}

In this section, we will use $\mathsf{NaivePerm}$
to generate P-samples in Algorithm~\ref{MainProtocol}.
Even though we see in Lemma~\ref{lemma:naive_attack}
that one P-sample from $\mathsf{NaivePerm}$ can easily be attacked,
the fact that violation can be detected means
that repeating enough number of P-samples can counter the effect
of the adversary's limited violation budget.
Another reason to have a larger number of P-samples is to
achieve high probability maximin security based on
measure concentration argument, such as the following variant of Chernoff Bound.

\begin{fact}[Chernoff Bound]
\label{fact:chernoff}
Suppose for each $j \in [R]$,
$Y_j$ is sampled independently from the same distribution with support $[0, U_{\max}]$
such that $\Gamma = \frac{U_{\max}}{\E[Y_j]}$.
Then, for any $0 < \epsilon < 1$,

$\Pr{\sum_{j \in [R]} Y_j \leq (1 - \epsilon) \cdot \E[\sum_{j \in [R]} Y_j]}
\leq \exp(-\frac{\epsilon^2 R}{2 \Gamma}).$
\end{fact}

This readily gives a guarantee on the reward of a player when the adversary does not interfere.
%The number of required P-samples to ensure $(\epsilon, \delta)$-maximin security is associated with the following parameters:

\begin{lemma}[Baseline Reward with No Adversary]
\label{lemma:baseline}
Suppose Algorithm~\ref{MainProtocol} is run
with $\mathsf{NaivePerm}$ and $R$ number of P-samples such that there is no interference from the adversary.
Then, for any $0 < \epsilon < 1$, the probability that
player~$i^*$ receives less than $(1 - \epsilon) \cdot \phi_{i^*}$
is at most $\exp(-\frac{\epsilon^2 R}{2 \Gamma})$,
which is at most $\delta$ when $R = \frac{2\Gamma}{\epsilon^2} \log \frac{1}{\delta}$.
\end{lemma}

\ignore{
\begin{proof}
For $j \in [R]$,
let $Y_j$ be the reward received by player~$i^*$ in the $j$-th P-sample, assuming that
the adversary does not interfere.
The result immediately follows from the Chernoff Bound in Fact~\ref{fact:chernoff},
with $\E[Y_j] = \phi_{i^*}$ and $U_{\max} = \Gamma \cdot \phi_{i^*}$ from the definition of \textit{max-to-mean} parameter $\Gamma$.
\end{proof}
}

Recall that in a supermodular game, $\Gamma \leq n$. By using this fact we can obtain a sampling complexity $O(n)$. We note that~\cite{sv_convex} has also analyzed the sampling complexity
with no adversary in supermodular games.  However, instead of using Chernoff Bound,
they used the weaker Chebyshev's Inequality.  As a result,
their bound has a factor of $O(n^2)$, as opposed to $O(n)$ in our bound.

\ignore{
By using this fact we can obtain a sampling complexity related to $n$ (Corollary~\ref{cor:sup-baseline}).
 \quan{Delete 4.3. Apply lemma 4.2 with Gamma = n}

\begin{corollary}[Baseline Reward in Supermodular Games]
	\label{cor:sup-baseline}
	Suppose Algorithm~\ref{MainProtocol} is run
	with $\mathsf{NaivePerm}$ and $R$ number of P-samples such that there is no interference from the adversary.
	Then, for any $0 < \epsilon < 1$, the probability that
	player~$i^*$ receives less than $(1 - \epsilon) \cdot \phi_{i^*}$
	is at most $\exp(-\frac{\epsilon^2 R}{2 n})$,
	which is at most $\delta$ when $R = \frac{2n}{\epsilon^2} \log \frac{1}{\delta}$.
\end{corollary}
}

\subsection{Warmup: Known Violation Budget}

\ignore{
While $\epsilon$-expected maximin security ensures expectation not to deviate too much, it is not guaranteed that the final allocation of players is close to the expectation sequence. In this section we deal with the advanced task, $(\epsilon, \delta)$-maximin security. Specifically, our task is to give new parameters to the stopping condition $\mathsf{stopcond}$ to make sure that the revenue of $i^*$ is at least $(1-\epsilon)\phii$ with probability at least $1-\delta$. This is similar to the analysis of sampling complexity of approximating the Shapley value by random permutations and we follow the similar idea.
}

The simple case is that the violation budget~$C$ is known in advance.
Because each violation can cause a damage of at most $U_{\max, i^*}  \leq \Gamma \cdot \phi_{i^*}$,
it is straightforward to give a high probability statement.
Recall that under the violation model,
each unit of the violation budget can cause a susceptible player
to violate the protocol once in a round.
(In this case,
actually it is not important whether a violation can be detected.)

\begin{theorem}[Known Violation Budget]
\label{UB_of_level2}
For any $0 < \epsilon, \delta < 1$,
when the violation budget~$C$ is known,
$(\epsilon, \delta)$-maximin security can be achieved
by setting the stopping condition $\mathsf{stopcond}$
as having the number of P-samples reaching:
$R = \max\{\frac{8\Gamma}{\epsilon^2} \ln \frac{1}{\delta}, \frac{2 C \Gamma}{\epsilon}\}.$
\end{theorem}

\begin{proof}
For $j \in [R]$, we decompose the reward received by player~$i^*$ in the $j$-th P-sample as $X_j := Y_j - Z_j$,
which is described in Remark~\ref{remark:decompose}.
Recall that $Y_j$ is the reward received by player~$i^*$
had the adversary not caused any violation (if any),
and $Z_j$ can be interpreted as the amount of damage
due to the adversary's violation in the $j$-th P-sample.

By Lemma~\ref{lemma:baseline}, $(1 - \frac{\epsilon}{2})$-approximation can be achieved
by $\sum_{j \in [R]} Y_j$
with probability at least $1 - \delta$, if $R \geq \frac{8\Gamma}{\epsilon^2} \ln \frac{1}{\delta}$.
In other words, under this range of $R$, with probability at least $1 - \delta$,
we have the lower bound for the sum $\sum_{j \in [R]} Y_j \geq (1 - \frac{\epsilon}{2}) \cdot R \phi_{i^*}$.

Now, the violation budget $C$ is known, which means that, with probability 1,
the adversary causes a total reward deduction  $\sum_{j \in [R]} Z_j \\ \leq C U_{\max, i^*} \leq C \Gamma \phi_{i^*}$,
which is at most $\frac{\epsilon}{2} \cdot R \phi_{i^*}$ if we choose $R \geq \frac{2 C \Gamma}{\epsilon}$.
Therefore, if $R = \max\{ \frac{8\Gamma}{\epsilon^2} \ln \frac{1}{\delta}, \frac{2 C \Gamma}{\epsilon} \}$,
$(\epsilon, \delta)$-maximin security is achieved.
\end{proof}

\ignore{
\begin{remark}
In Theorem~\ref{UB_of_level2},
we see that the sampling complexity $R$ depends on
two quantities: $O(\frac{\Gamma}{\epsilon^2} \ln \frac{1}{\delta})$
and $O(\frac{C\Gamma}{\epsilon})$.
While the first quantity comes from the measure concentration
analysis in Lemma~\ref{lemma:baseline},
we shall investigate in Section~\ref{sec:lowerbound}
that the second quantity is necessary in our P-sample model,
even when we try to use a more ``secure'' P-sample protocol
to replace $\mathsf{NaivePerm}$.
\end{remark}
}

\subsection{Unknown Violation Budget}

In the case where violation budget is unknown, the stopping condition
in Algorithm~\ref{MainProtocol} may no longer specify
the pre-determined number of P-samples.
However, one advantage of $\mathsf{NaivePerm}$
is that it can detect whether a violation has occurred.
Therefore, we can use the fraction of P-samples
in which violation has occurred in the stopping condition.

Note that if the adversary has an unlimited violation budget,
then the protocol may never terminate.
Hence, in the most general case,
the following lemma does not directly guarantee
$(\epsilon, \delta)$-maximin security.

\begin{lemma}[General Unknown Budget]
\label{UB_of_unlimited}
Suppose the adversary has some unknown violation budget (that may be unlimited).
Furthermore, for $0 < \epsilon, \delta < 1$, Algorithm~\ref{MainProtocol}
is run with $\mathsf{NaivePerm}$ and the
stopping condition $\mathsf{stopcond}$ is:
\begin{itemize}
\item the fraction of P-samples observed so far
with detected violations is at most $\frac{\epsilon}{2\Gamma}$; AND

\item the number of P-samples has reached at least:

$R_0 := \frac{8\Gamma}{\epsilon^2} \cdot (\ln \frac{16\Gamma}{\epsilon^2} + \ln \frac{1}{\delta} )$.

\end{itemize}

%Let the stopping condition $\mathsf{stopcond}$ be "reaching $R_0$ P-samples and the fraction of honest P-samples over total P-samples is at least $F$" where
    %\begin{equation}
        %\begin{aligned}
        %F &= 1-\frac{\epsilon}{2n}, \\
        %R_0 &= -\ln\left[ \delta \left( 1-\exp\left( -\frac{\epsilon^2}{12n} \right) \right) \right] \cdot \frac{12n}{\epsilon^2}.
        %\end{aligned}
    %\end{equation}

    Then, the probability that the protocol terminates and player~$i^*$ receives an outcome of less than $(1-\epsilon)\phii$ is at most $\delta$.
\end{lemma}

\begin{proof}
For $R \geq R_0$,
let $\delta_R$ be the  probability of the
bad event $\mcal{E}_R$ that the protocol terminates after exactly $R$ number
of P-samples and the honest player~$i^*$
receives an outcome of less than $(1-\epsilon)\phii$.
Our final goal is to show that $\sum_{R \geq R_0} \delta_R \leq \delta$.

\noindent \textbf{Analyzing Fixed $\mcal{E}_R$.}
We fix some $R$ and analyze $\delta_R$.
Recall that for $j \in [R]$,
we consider the reward decomposition $X_j := Y_j - Z_j$
as described in Remark~\ref{remark:decompose}.
Observe that the bad event $\mcal{E}_R$ implies both of the following
have occurred:
\begin{itemize}
\item  The number of P-samples with detected violation is at
most~$\frac{\epsilon R}{2 \Gamma}$.
Therefore, $\sum_{j \in [R]} Z_j \leq
\frac{\epsilon R}{2 \Gamma} \cdot U_{\max, i^*}
\leq \frac{\epsilon}{2} \cdot R \phii$.

\item  Had the adversary not interfered,
the sum is:

 $\sum_{j \in [R]} Y_j < (1 - \frac{\epsilon}{2}) R  \cdot \phii$.
\end{itemize}

The first point is part of the stopping condition.
If the second point does not happen,
this means that $\sum_{j \in [R]} (Y_j - Z_j)\geq
(1 - \frac{\epsilon}{2}) R  \cdot \phii
- \frac{\epsilon}{2} \cdot R \phii  =
(1 - \epsilon) R  \cdot \phii$, thereby violating $\mcal{E}_R$.

Therefore, we have $\delta_R \leq \Pr{\sum_{j \in [R]} Y_j < (1 - \frac{\epsilon}{2}) R  \cdot \phii}  \leq \exp(- \frac{\epsilon^2 R}{8\Gamma})$,
where the last inequality follows
from Lemma~\ref{lemma:baseline}.

\noindent \textbf{Analyzing the Union of All Bad Events.}
Therefore, we have

$$\sum_{R \geq R_0} \delta_R \leq
\frac{\exp\left( - \frac{\epsilon^2 R_0}{8\Gamma} \right)}{1 - \exp\left( - \frac{\epsilon^2}{8\Gamma} \right)} \leq \frac{16 \Gamma}{\epsilon^2 } \exp\left( - \frac{\epsilon^2 R_0}{8\Gamma} \right),$$

where the last inequality follows
because $1 - e^{-t} \geq \frac{t}{2}$ for $t \in (0,1)$.
By choosing $R_0 := \frac{8\Gamma}{\epsilon^2} \cdot (\ln \frac{16\Gamma}{\epsilon^2} + \ln \frac{1}{\delta} )$, the last expression equals $\delta$.  This completes the proof.
\end{proof}

\noindent \textbf{Finite Violation Budget.}
The adversary has some finite violation budget $C$
that is unknown to the protocol.
The stopping condition ensures that the protocol will terminate with probability 1.
In this case, $(\epsilon, \delta)$-maximin security follows immediately.

\begin{corollary}[Maximin Security for Finite Unknown Budget]
\label{cor:finite_unknown}
Suppose in Lemma~\ref{UB_of_unlimited},
the adversary has some finite (but unknown) violation budget $C$.
Let $0 < \epsilon, \delta < 1$ and $R_0$ be as defined in Lemma~\ref{UB_of_unlimited}.
Then, with probability 1,
the protocol terminates after at most
$\max\{R_0, \frac{2C\Gamma}{\epsilon}\}$ number of P-samples.
Moreover, $(\epsilon, \delta)$-maximin security is achieved for
an honest player.
\end{corollary}

\subsection{Unknown Violation Rate}

The violation rate~$f \in [0,1]$ means that
for every $T > 0$, during the generation of the first $T$ number of P-samples,
there can be at most $f T$ violations. However,
the protocol does not know $f$ in advance.

\noindent \textbf{Adaptive Maximin Security.}
The protocol has some target approximation parameter $\epsilon \in [0,1]$.
However, whether this is achievable depends on the violation rate~$f$
and the max-to-mean ratio~$\Gamma$.
\ignore{
In the scenario where the violation budget is unbounded, if we use the stopping condition described in Lemma~\ref{UB_of_unlimited}, it is possible that the protocol never terminates. This will occur even when the violation budget is limited by a fraction of the number of P-samples, that is, at any moment when the adversary has executed $T$ P-samples, the adversary has a budget $C = fT$, where $f$ represents the fraction. If $\frac{\epsilon}{2\Gamma} < f$ and the adversary violates whenever it can, then the stopping condition in Lemma~\ref{UB_of_unlimited} will never be satisfied, leading to the impossibility of achieving $(\epsilon, \delta)$-maximin security.

Since some $(\epsilon, \delta)$-maximin security is not attainable, it's reasonable to adjust $\epsilon$ and the stopping condition according to the observation so that the protocol always terminates with an attainable $\epsilon$.
}

Algorithm~\ref{adprotocol} outlines an adaptive algorithm that returns both an allocation $x$ and an approximation parameter $\widehat{\epsilon}$. Theorem~\ref{lemma:adaptive_rate} describes the security property of Algorithm~\ref{adprotocol}.

\begin{algorithm}
	\caption{Abstract Adaptive Protocol to Return an Allocation}\label{adprotocol}
	\SetKwFunction{GenPerm}{GenPerm}
	\KwIn{A game $(\all, v)$, $\Gamma$, $\delta$, $\epsilon$.}
	\KwOut{An allocation $x \in \R_+^\all$ and $\epsilon_k$.}
	%$N \gets \all$ \\
	$z, x \gets \vec{0} \in \R_+^\all$,  $R \gets 0$ \\
	%define macro $\epsilon_t = 2^{-t}$
	%$\epsilon_t$\Comment{It always holds that $\epsilon_t = 2^{-t}$.} %$\xi$ is created to make the notation cleaner.  %$\xi \gets \frac{\epsilon}{2}$
	$k \gets 0$,
	$V \gets 0$\\ %\Comment{The total number of violations.}
	\Comment{For all $t \geq 0$, denote $\epsilon_t := 2^{-t}$.}
	%$ \mathsf{stopcond} = \mathsf{false}$\\
	\While{$\epsilon_k > \epsilon$}{
		\Comment{Here $\aux$ is the number of violations.}
		$(\sigma, \aux) \gets$ \GenPerm{$\all$} \\
		%\Comment{$\sigma_{\all \backslash N}$ is an arbitrary permutation of players in $\all \backslash N$. }
		%$\sigma \leftarrow \sigma  \bigoplus \sigma_{\all \backslash N}$
		\Comment{Marginal contribution of rank~$j$ player in $\sigma$}
		\textbf{for every} $j \in [n]$,
		$z_{\sigma_j} \gets z_{\sigma_j} + \mu_{\sigma_j}(\sigma_{<j})$ \\
		$R \gets R+1$ \\
		%\Comment{$\aux$ stores the information about the violations in this iteration}
		$V \gets V + |\aux|$\\
		%\Comment{May remove players from $N$ according to $\aux$}
		\If{$R \geq \frac{8\Gamma}{\epsilon_{k+1}^2}\ln \frac{2^{k+1}}{\delta}$}{
			\If{$\frac{V}{R} > \frac{\epsilon_{k+1}}{2\Gamma}$}{
				%$\epsilon \gets 2\epsilon$\\
				\textbf{break}}
			\Else{
				%$\epsilon \gets \frac{\epsilon}{2}$,
				%$\xi \gets \frac{\xi}{2}$\\
				$x \gets \frac{1}{R} \cdot z$\\
				$k \gets k + 1$\\
				%$\epsilon \gets \epsilon_k$
			}
		}

	}
	\Return $(x, \epsilon_k)$
	\Comment{For any honest player $i^*$, with probability at most $\delta$, $x_{i^*} < (1-\epsilon_k)\phii$. }
\end{algorithm}

\begin{theorem}[Adaptive Maximin Security for Violation Rate]
	\label{lemma:adaptive_rate}
	Suppose the adversary has some (unknown) violation rate $f \in [0,1]$.
	%
	%some unknown violation budget $C = f \cdot T$ (that may be unlimited), where where $T$ is the number of P-samples and $0 \leq f \leq 1$.
	%
		Furthermore, for any $0 < \epsilon, \delta < 1$, we execute Algorithm~\ref{adprotocol} using the $\mathsf{NaivePerm}$ approach.
		%, aiming to achieve $(\epsilon, \delta)$-maximin security.
	
	Then, the protocol
	returns $(x, \widehat{\epsilon})$ such that
	$\widehat{\epsilon} \leq \max\{\epsilon, 4f\Gamma\}$ holds with probability 1.
	Moreover, for any honest player $i^*$,
	with probability at least $1 - \delta$,
	its received reward satisfies $x_{i^*} \geq (1- \widehat{\epsilon})\phii$.
	
	%Then, the protocol terminates with a probability of 1. Player $i^*$ receives a reward of less than $(1- \epsilon)\phii$ with a probability of at most $\delta$, where $\epsilon$ represents the value of the returned variable $\epsilon$. Moreover, either $\epsilon = \hat{\epsilon}$ or $\epsilon < 4f\Gamma$ holds. %\quan{How about use $\hat{\epsilon}$ instead??}
	%Then, the probability that the protocol terminates and player~$i^*$ receives an outcome of less than $(1-\epsilon)\phii$ is at most $\delta$.
\end{theorem}

\begin{proof}

	Let $(x, \epsilon_k)$ be the returned value. The algorithm terminates either when $\epsilon_k \leq \epsilon$ or $\frac{V}{R} > \frac{\epsilon_{k+1}}{2\Gamma}$. Since $f \geq \frac{V}{R}$ and $\epsilon_{k+1} = \frac{\epsilon_k}{2}$, the second inequality implies that $\epsilon_k < 4f\Gamma$.
	
	Let $\mcal{B}$ denote the bad event that the protocol terminates and player $i^*$ receives a reward of less than $(1- \epsilon_k)\phii$. Let $\epsilon_t = \frac{1}{2^t}$. Let $\mcal{B}_t$ denote the event that the protocol terminates with the value of variable $k$ equals $t$, and player $i^*$ receives a reward of less than $(1- \epsilon_t)\phii$.
    When $t = 0$, $\epsilon_t = 1$. Since $i^*$ will receive a non-negative reward, $\Pr{\mcal{B}_0} = 0$, so we only focus on the cases when $t \geq 1$.
	It holds that $\mcal{B} \subseteq \cup_{t \ge 1} \mcal{B}_t$, therefore, $\Pr{\mcal{B}} \leq \Pr{\cup_{t \ge 1} \mcal{B}_t}$. Our goal is to show that:
	\begin{equation*}
    \Pr{\cup_{t \ge 1} \mcal{B}_t} \leq \delta.
    \end{equation*}
	
	\noindent \textbf{Analyzing the probability of $\cup_{t \ge 1} \mcal{B}_t$. }
	We can follow a similar argument as in the Lemma~\ref{UB_of_unlimited} about analyzing $\mcal{E}_R$.
	Let $R_t = \frac{8\Gamma}{\epsilon_t^2} \ln \frac{2^{t}}{\delta}$. %\quan{Actually I think $R_t = \lceil \frac{8\Gamma}{(\epsilon_t)^2} \cdot (\ln \frac{2^{t}}{\delta} ) \rceil$ is enough.}
	Let $\mcal{B}_t'$ denote the event that when $\mathsf{NaivePerm}$ has run $R_t$ number of P-samples and the fraction of P-samples observed so far with detected violations is at most $\frac{\epsilon_t}{2\Gamma}$, player $i^*$ receives a reward of less than $(1- \epsilon_t)\phii$. It's clear that $\mcal{B}_t \subseteq \mcal{B}_t'$, so $\Pr{\mcal{B}_t} \leq \Pr{\mcal{B}_t'}$.
	
	Recall that for $j \in [R]$,
	we consider the reward decomposition $X_j := Y_j - Z_j$
	as described in Remark~\ref{remark:decompose}.
	Observe that the event $\mcal{B}_t'$ implies both of the following
	have occurred:
	\begin{itemize}
		\item  The number of P-samples with detected violation is at
		most~$\frac{\epsilon_t R_t}{2 \Gamma}$.
				Therefore, $\sum_{j \in [R_t]} Z_j \leq
		\frac{\epsilon_t R_t}{2 \Gamma} \cdot U_{\max, i^*}
		\leq \frac{\epsilon_t}{2} \cdot R_t \phii$.
		
		\item  Had the adversary not interfered,
		the sum  is:
		
		$\sum_{j \in [R_t]} Y_j < (1 - \frac{\epsilon_t}{2}) R_t  \cdot \phii$.
	\end{itemize}
	
	The first point is from the definition of $\mcal{B}_t'$.
	If the second point does not happen,
	this means that $\sum_{j \in [R_t]} (Y_j - Z_j)\geq
	(1 - \frac{\epsilon_t }{2}) R_t  \cdot \phii
	- \frac{\epsilon_t}{2} \cdot R_t \phii  =
	(1 - \epsilon_t) R_t  \cdot \phii$, thereby violating $\mcal{B}_t'$.

	Therefore, we have $\Pr{\mcal{B}_t'} \leq \Pr{\sum_{j \in [R_t]} Y_j < (1 - \frac{\epsilon_t}{2}) R_t  \cdot \phii}  \leq \exp(- \frac{\epsilon_t^2 R_t}{8\Gamma}) = \frac{\delta}{2^t}$,
	where the last inequality follows
	from Lemma~\ref{lemma:baseline}.
	
	%\noindent \textbf{Analyzing the the probability of $\mcal{B}$.}
	
	Therefore, we have
	
	$$\Pr{\mcal{B}} \leq \Pr{\cup_{t \ge 1} \mcal{B}_t} \leq \Pr{\cup_{t \ge 1} \mcal{B}_t'} \leq \sum_{t \ge 1} \frac{\delta}{2^t} \leq \delta, $$
	where the first inequality follows because $\mcal{B} \subseteq \cup_{t \ge 1} \mcal{B}_t$, the second inequality holds because $\mcal{B}_t \subseteq \mcal{B}_t'$, the third inequality is from the union bound.
\end{proof}

%\begin{proof}
%With unknown but limited budget, the protocol always terminates because once the budget runs out, the fraction and the number of P-samples can only increase. Suppose under some strategy of adversaries the $(\epsilon, \delta)$-maximin security does not hold when the protocol terminates, the adversary can apply the same strategy to the case of unlimited budget, causing that the protocol terminates and $i^*$ fails to gain $(1-\epsilon)\phii$ with probability at least $\delta$, which is contradicted to Lemma \ref{UB_of_unlimited}.
%\end{proof} 

%% file: expectation.tex
\section{Sampling Complexity to Achieve Expected Maximin Security}
\label{sec:lowerbound}

In \Cref{algorithm},
we give an abstract Algorithm~\ref{MainProtocol}
that is based on permutation sampling.
Using $\mathsf{NaivePerm}$ to generate P-samples,
we see in \Cref{sec:high_prob} that to achieve $(\epsilon, \delta)$-maximin security,
a sampling complexity of
$\max \{O(\frac{\Gamma}{\epsilon^2} \ln \frac{1}{\delta}),
O(\frac{C\Gamma}{\epsilon}) \}$ is sufficient.
While the term $O(\frac{\Gamma}{\epsilon^2} \ln \frac{1}{\delta})$ is due
to measure concentration analysis to achieve a high probability statement,
in this section, we shall
investigate that term $O(\frac{C\Gamma}{\epsilon})$
is actually sufficient and necessary
to achieve $\epsilon$-expected maximin security
by Algorithm~\ref{MainProtocol},
even when we try to replace $\mathsf{NaivePerm}$
with a more ``secure'' P-sample protocol in supermodular games. %\quan{This is not accurate. We not only replace naiveperm, but also use perpetual violation model. }

While we give an upper bound of the sampling complexity
to achieve expected maximin security for the baseline protocol using $\mathsf{NaivePerm}$,
the main result of this section
is a lower bound on the sampling complexity
even when we restrict to supermodular games and impose further conditions to help an honest player.

%Moreover, we suggested two distributed protocols,
%$\mathsf{NaivePerm}$ and $\mathsf{SeqPerm}$,
%to generate a random permutation.
%Even though we give some guarantees for $\mathsf{SeqPerm}$ in
%Claim~\ref{claim:seqperm},
%it is somehow surprising that $\mathsf{NaivePerm}$
%can achieve similar sampling complexity
%to achieve expected maximin security.

\noindent \textbf{Model Assumptions.}  We
consider an adversary with some known budget~$C$.
%The intuition is that if we repeat enough number of P-samples,
%then most of them will not be interfered by the adversary,
%thereby giving some guarantee for the reward of an honest player~$i^*$.
For analyzing the upper bound of sampling complexity,
we consider using the same model as in \Cref{sec:high_prob}.

For analyzing the lower bound of sampling complexity,
we consider supermodular games,
and replace $\mathsf{NaivePerm}$ with another P-sample protocol
known as $\mathsf{SeqPerm}$ that has better security properties.
Moreover, we consider the perpetual violation model
in which detected violating players will be moved
to the least preferable positions, thereby giving an advantage
to an honest player.

\subsection{Upper Bound for Number of P-Samples Using $\mathsf{NaivePerm}$}

\noindent \textbf{Baseline Protocol.}  We instantiate
the abstract protocol in Algorithm~\ref{MainProtocol}
by using $\mathsf{NaivePerm}$ to sample permutations.
Moreover, the stopping condition can simply be having
enough number of P-samples.

\begin{lemma}[Upper Bound on Sampling Complexity]
\label{lemma:ub_naiveperm}
By using $\mathsf{NaivePerm}$ in Algorithm~\ref{MainProtocol}
against an adversary with a known violation budget of~$C$,
for any $\epsilon > 0$,
$R = \frac{\Gamma C}{\epsilon}$ number of P-samples
is sufficient to achieve $\epsilon$-expected maximin security.
\end{lemma}

\begin{proof}
Consider running Algorithm~\ref{MainProtocol} with
$R = \frac{\Gamma C}{\epsilon}$ number of P-samples.
For $j \in [R]$,
we consider the reward decomposition
$X_j := Y_j - Z_j$ as described in Remark~\ref{remark:decompose}.
When the adversary does not interfere,
we have $\sum_{j \in [R]} \E[Y_j] = R \cdot \phi_{i^*}$.

%Each P-sample corresponds
%to a random variable $U$ denoting the reward
%of some honest player~$i^*$ derived from that random permutation.
%As observed in Fact~\ref{fact:max_ratio},
%$U$ has support in $[0..U_{\max}]$, where $U_{\max} \leq n \cdot \phi_{i^*} = n \cdot \E[U]$.

%Consider running Algorithm~\ref{MainProtocol} with
%$R = \frac{n C}{\epsilon}$ number of P-samples.
Because of the violation budget~$C$,
the adversary can only instruct susceptible players
to violate in at most $C$ of the P-samples.
Hence, with probability 1, $\sum_{j \in [R]} Z_j \leq
C \cdot U_{\max, i^*} \leq C\Gamma \cdot \phi_{i^*}$.
%where violation in each P-sample can be viewed as deducting at most $U_{\max}$ from the honest player's reward.
%Next, observe that if the adversary
%did not instruct susceptible players to violate the protocol,
%the expected reward from $R$ samples
%is exactly $R \cdot \phi_{i^*}$.
Therefore, the actual
reward received from the $R$ samples
is at least $R \cdot \phi_{i^*} - \Gamma C \cdot \phi_{i^*}
= R \cdot (1 - \epsilon) \cdot \phi_{i^*}$.
Dividing this by $R$, we see that the expected outcome for
player~$i^*$ is at least
$(1 - \epsilon) \cdot \phi_{i^*}$.
\end{proof}

\subsection{Punish Violating Players in Supermodular Games}
In this section, we demonstrate that in Algorithm~\ref{MainProtocol}, if the violating players are given the least preferable positions, then the expected reward of an honest player $i^*$ will not decrease in supermodular games. We call this violation model to be \textit{perpetual}.

%we show Algorithm~\ref{MainProtocol} becomes more ``secure" if violating players are removed from $N$. \quan{Can we say it is more secure??}

%in by removing violated players from $N$ in

\noindent \textbf{Perpetual Violation. }In the random permutation approach, detected violation by a player~$i$
can be handled by always placing that player~$i$ in the least preferable rank
in future permutation samplings.
This can be considered as a modified game $v': 2^{\all \setminus \{i\}} \rightarrow \R$
with one less player, where $v'(S) = v(S \cup \{i\})$.
The following lemma
states that this cannot decrease the expected reward of an honest player~$i^*$ in supermodular games.

\begin{lemma}[Punishing Perpetual Violating Players]
	\label{lemma:perpetual}
	Consider sampling a random permutation
	uniformly at random from some restricted set such that
	some player $i \neq i^*$ is always at the least preferable rank.
	Then, the expected reward of player~$i^*$
	is still at least $\phi_{i^*}$.
\end{lemma}

\begin{proof}
	Let $\phi'_{i^*}$ be the expected reward
	of player~$i^*$ under the restricted sampling procedure. Then, we have
	%
	%Suppose a set of players $T$ have been detected corrupting before this P-sample. A new set of players planning to corrupt can be regarded as corrupting one by one. Therefore, we only consider the case where only one player corrupts, which can be applied inductively to the general case. Suppose that player $j$ corrupts, after which the expected revenue of $i^*$ is $E'$, then
	\begin{equation*}
		\begin{aligned}
			\phi'_{i^*}-\phi_{i^*} =& \sum_{S \subseteq \all \setminus \{i^*,i\}} \frac{|S|!(n-|S|-2)!}{(n-1)!} \cdot \mu_{i^*}(S \cup \{i\}) - \\ &\sum_{S \subseteq \all \setminus \{i^*\}} \frac{|S|!(n-|S|-1)!}{n!} \cdot \mu_{i^*}(S) \\
			=& \sum_{S \subseteq \all \setminus \{i^*,i\}} \frac{|S|!(n-|S|-2)!}{(n-1)!} \cdot \mu_{i^*}(S \cup \{i\}) \\
			&- \sum_{S \subseteq \all \setminus \{i^*,i\}}  (\frac{(|S|+1)!(n-|S|-2)!}{n!} \cdot \mu_{i^*}(S\cup\{i\}) + \\ & \frac{|S|!(n-|S|-1)!}{n!} \cdot \mu_{i^*}(S) ) \\
			=& \sum_{S \subseteq \all \setminus \{i^*,i\}} \frac{|S|!(n-|S|-1)!}{n!} \cdot (\mu_{i^*}(S\cup\{i\}) - \mu_{i^*}(S)) \geq 0,
		\end{aligned}
	\end{equation*}
	
	where the last inequality follows from Fact \ref{MMC}.
	%, $\mu_{i^*}(S\cup\{j\}) \ge \mu_{i^*}(S)$ for every $S \subseteq \all \setminus \{i^*,i\}$, hence
	%\begin{equation}
	%E'-E \ge 0.
	%\end{equation}
\end{proof}

\begin{remark}
	If more than one player is detected to be violating, Lemma \ref{lemma:perpetual} can be applied inductively to show the same guarantee of the expected reward of player $i^*$.
\end{remark}

\ignore{
\begin{remark}
	It can be inferred from Lemma \ref{lemma:perpetual} that for any adversary strategy under perpetual violation, there exists an adversary strategy under normal violation, where the reward of player $i^*$ is lower. \quan{Maybe need more illustration. }
\end{remark}
}

Corollary~\ref{cor:removeN} follows directly from Lemma~\ref{lemma:perpetual}.
\begin{corollary}
	\label{cor:removeN}
	In Algorithm~\ref{MainProtocol}, if violating players are given least
	preferable positions, then the expected reward of an honest player~$i^*$ will not decrease in supermodular games.
\end{corollary}
%In other word, in Algorithm~\ref{MainProtocol},  if $i$ violates the protocol,

\subsection{Replacing $\mathsf{NaivePerm}$ with a More ``Secure'' P-Sample Protocol in Supermodular Games}
\label{sec:seqperm}

\ignore{
\noindent \textbf{Perpetual Violation.}
In the random permutation approach, detected violation by a player~$i$
can be handled by always placing that player~$i$ in the least preferable rank
in future permutation samplings.
This can be considered as a modified game $v': 2^{\all \setminus \{i\}} \rightarrow \R$
with one less player, where $v'(S) = v(S \cup \{i\})$.
The following lemma
states that this cannot decrease the expected reward of an honest player~$i^*$.}

Lemma~\ref{lemma:naive_attack} states that if the game is supermodular,
in $\mathsf{NaivePerm}$,
an honest player can be forced to go to the least (or the second least) preferable
position in the permutation with probability 1
if the adversary just uses one unit of the violation budget.
To rectify this situation, a simple idea
is to generate a random permutation by deciding which player
goes to which position sequentially, starting from the least preferable position.
Starting from an \emph{active} set of players, a variant of Blum's protocol~\cite{tele_coin_flipping} can eliminate
a uniformly random player from the active set one at a time, which goes to
least preferable available remaining position.

%The idea is to sample sufficient permutations and get the approximation of Shapley value by the average of marginal contributions of each player, according to Equation \ref{sv}. We call the procedure of sampling a permutation a \emph{P-sample}. The protocol can be transfer in to a non-adaptive version and an adaptive version depending on different stopping conditions $STP$. If $STP$ is just "reaching a certain number of P-samples", it is a non-adaptive version. If $STP$ depends on the fraction of permutations with abort and without abort, for example, it is an adaptive version.

\ignore{
In this section we describe our protocol satisfying the conditions in section \ref{SecureSetting}, and give analysis on it in the next section. Our protocol is based on the approximation of Shapley value by sampling random permutations. Specifically, players decide the random permutation from the least preferable position to the most preferable position. The player lying in each position is decided by a simple game of producing a random index over the remaining players. Repeating the procedure for a specified number of rounds, the mean of marginal contribution of each player is his allocated revenue.
}

\ignore{
Note that there might be other methods to approximate Shapley value for each player more efficiently, we would not use them since we need to obtain an allocation at the end of the protocol while approximating the Shapley value for each player separately does not guarantee to produce an allocation (the sum of approximated value might not be $v(\all)$).
}

Algorithm~\ref{RandElim} is a straightforward generalization
of the well-known Blum's protocol to more than two players.

\begin{algorithm}
\caption{Random Eliminate $\mathsf{RandElim}$}\label{RandElim}
\SetKw{KwDownTo}{downto}
\KwIn{A player set $S$.}
\KwOut{A player $a \in S$ that is eliminated,
together with a collection $\mathsf{Dev} \subseteq S$ of detected violating players.}

Let $k = |S|$, and consider some linear order of players in $S = \{a_0, \ldots , a_{k-1}\}$, e.g., the lexicographical order. \\

\Comment{Commit phase:}
\ForEach{{\rm \bf player} $p \in S$} {
    Uniformly sample an integer $r_p \in [0..k-1]$ \\
    Commit $r_p$ and broadcast the commitment. \\
}

\Comment{Open phase:}
\ForEach{{\rm \bf player} $p \in S$} {
		Broadcast the opening for its previously committed $r_p$.

}

Let $\mathsf{Dev}$ denote the collection of players that
do not follow the commit or the open phase,
i.e., such a player does not broadcast an opening
that can be verified with its previous commitment.

\If{$\mathsf{Dev} = \emptyset$} {
    $i \gets (\sum_{p \in S} r_p) \bmod k$ \\
} \Else {
    $i \gets$ index of any arbitrary player in $\mathsf{Dev}$ (e.g., smallest index) \\
}
\Return $(a_i, \mathsf{Dev})$
\end{algorithm}

\begin{claim}[Probability of Elimination]
\label{claim:prob_elim}
Assume that the commitment scheme in Algorithm~\ref{RandElim} is ideal.
When $\mathsf{RandElim}$ is run between a collection $S$ of players,
an honest player~$i^* \in S$ is chosen to be eliminated
with probability at most $\frac{1}{|S|}$,
even if a Byzantine adversary controls all other players.
\end{claim}

\begin{proof}
Since the argument is essentially the same as the well-known Blum's protocol~\cite{tele_coin_flipping},
we only give a brief intuition.  Because any player not finishing the commit and the open phases
will imply that the honest player~$i^*$ is not eliminated,
an adversary cannot increase the chance of eliminating~$i^*$ by causing a violating player to be detected.

Assuming that the commitment scheme is ideal,
the adversary gets no extra information from the commitment of the honest player's random number
which is sampled uniformly at random.
Therefore, no matter how the susceptible players sample their random numbers,
the resulting sum modulo $|S|$ is still uniformly random.
\end{proof}

\noindent \textbf{Sequential Permutation Generation.}
We can generate a random permutation sequentially in $n$ phases,
where in each phase,
the protocol $\mathsf{RandElim}$ is used to determine which player
goes to the least preferable available position.
Observe that if more than one violating player is detected in $\mathsf{RandElim}$,
they will go straight to the least preferable available positions in a batch.
Claim~\ref{claim:seqperm} is a direct consequence of Claim~\ref{claim:prob_elim}.

\ignore{
The protocol consists of $n$ rounds, each of which decides a player from the least preferable position to the most preferable position. By Fact \ref{MMC}, every current position is the least preferable among all the to-be-decide positions for every players. Therefore, to avoid repeating the game infinitely when abort appears as well as a punishment, if someone aborts in a round (e.g. refuse to open his commitment), he would fall into the current least preferable position. If more than one player abort, they would go to these bad positions in the order of their indices.

Next we show the protocol of generating a random permutation by the simple game, as Algorithm \ref{RandPerm}.
}

\begin{algorithm}
\caption{Sequential Permutation Generation $\mathsf{SeqPerm}$}\label{SeqPerm}
\SetKw{KwDownTo}{downto}
\SetKwFunction{RandElim}{RandElim}
\KwIn{Player set $N$.}
\KwOut{A random permutation $\sigma = (\sigma_1, \ldots, \sigma_n)$ of $N$,
together with a collection $\mathsf{Dev} \subseteq N$ of detected violating players.}

$S \gets \all$, $i \gets 1$ \\
\While {$S \neq \emptyset$} {
    $(a, \mathsf{Dev}) \gets$ \RandElim{$S$} \label{ln:RandElim} \\
    \If {$\mathsf{Dev} \neq \emptyset$} {
        $\sigma_i, \cdots, \sigma_{i+|\mathsf{Dev}|-1} \gets$ players in $\mathsf{Dev}$ in lexicographical order \\
        $i \gets i+|\mathsf{Dev}|$
    } \Else {
        $\sigma_i \gets a$ \\
        $i \gets i+1$
    }
    $S \gets S \setminus (\mathsf{Dev} \cup \{a\})$
}
\Return $(\sigma, \mathsf{Dev})$
\end{algorithm}

\begin{claim}[Guarantee for Preferable Positions]
\label{claim:seqperm}
In the random permutation generated by Algorithm~\ref{SeqPerm},
for any $k \in [n]$,
an honest player will be in one of the top $k$ most preferable positions with
probability at least $\frac{k}{n}$,
even if a Byzantine adversary controls all other players.
\end{claim}

\noindent \textbf{$\mathsf{NaivePerm}$ vs $\mathsf{SeqPerm}$.}
It might seem that $\mathsf{SeqPerm}$
generate a more ``secure'' random permutation
than $\mathsf{NaivePerm}$.  Hence, it would be somehow surprising
that, in \Cref{sec:lb_seqperm},
we show that even when we replace
$\mathsf{NaivePerm}$ with $\mathsf{SeqPerm}$,
there is a lower bound on the sampling complexity
that asymptotically matches the baseline in Lemma~\ref{lemma:ub_naiveperm}.

%that when one decides to use which permutation generation
%method in Algorithm~\ref{GenPerm}, it would later turn out that
%$\mathsf{SeqPerm}$ does not have a huge advantage over
%$\mathsf{NaivePerm}$ when one wishes to achieve $(\epsilon, \delta)$-maximin security.
%In Section~\ref{??????}, we show an attack strategy when one uses
%$\mathsf{SeqPerm}$ under the perpetual violation model.
%In Section~\ref{????????}, we show that using $\mathsf{NaivePerm}$
%is actually sufficient even in the transient violation model.

\subsection{Lower Bound for Number of P-Samples Using $\mathsf{SeqPerm}$ in Supermodular Games}
\label{sec:lb_seqperm}

\noindent \textbf{How $\mathsf{SeqPerm}$ can be attacked.}
Even though Claim~\ref{claim:seqperm} guarantees the probability
that an honest player~$i^*$ will be in preferable positions,
susceptible players can still hurt player~$i^*$ by
trading positions when a permutation is sampled.

Consider the following (supermodular) game $v: 2^\all \rightarrow \R_+$ in which
there are two special players $i^*$ and $j^*$,
where $i^*$ is honest and all other players are susceptible.

$$
v(S) = \begin{cases}
2, & \{i^*, j^*\} \subseteq S \\
0, & \text{otherwise,}
\end{cases}
$$

Observe that in a uniformly random permutation,
$i^*$ obtains a non-zero reward (of 2) \emph{iff}
$i^*$ is in a more preferable position than $j^*$,
which happens with probability $\frac{1}{2}$.
Hence, for this game, the expected reward
in one P-sample is $\phi_{i^*} = 1$.

However, when $\mathsf{SeqPerm}$ is used to
generate a P-sample, the adversary uses the following attack strategy.
During each iteration of $\mathsf{RandElim}$ in line~\ref{ln:RandElim}
of Algorithm~\ref{SeqPerm}, after $i^*$ opens its committed random number,
if the adversary sees that $j^*$ is going to be eliminated,
then it instructs another player to sacrifice itself by aborting.
Under this attack strategy, player~$i^*$ can get a non-zero reward
with probability $\frac{1}{n}$.

Therefore, if only 1 P-sample is used and the adversary has $n-1$ violation budget,
then the expected reward of $i^*$ is decreased from $1$ to
$\frac{2}{n}$.

\ignore{
It is easy to see that the protocol in section \ref{algorithm} satisfies all requirements except for the $\epsilon$-maximin security. In this section we analyze the sampling complexity to achieve $\epsilon$-maximin security under different corruption levels. Specifically, we specify the parameters in the stopping conditions of Algorithm \ref{MainProtocol}, i.e. how many P-samples do we need, or reaching what fraction should we stop the protocol. We give analysis for both upper bounds and lower bounds to illustrate that our protocol achieves nearly the optimal number of sampling of permutations. Unless there is a breakthrough in approximation methods that are not based on sampling, our method achieves nearly the optimal sampling complexity to the utility function as well.

Assume that among the $n$ players only player $i^*$ is honest while other $n-1$ players are adversaries with the goal of reducing the revenue of player $i^*$ as much as possible, which is the worst case for honest players.

\subsection{Properties in Level 1}
}

The above example gives an intuition of how the adversary might attack
an honest player.  However, under the perpetual violation model,
any detected violating player will be removed by the protocol, thereby possibly
making the job of the adversary more difficult (see Lemma~\ref{lemma:perpetual}).
Certainly if the protocol uses a large number of P-samples,
the adversary will eventually run out of susceptible players.
The following theorem gives a lower bound on the number
of P-samples needed to achieve expected maximin security
even under the perpetual violation model.

%which can cause significant harm to the honest player's outcome, resulting in a loss of approximately $\Theta(x \cdot v_{max})$.

%constructing a game and a strategy for adversaries,  (lemma~\ref{lem:xabnoregen}).

\begin{theorem}
\label{LB_of_Level1}
Suppose $\mathsf{SeqPerm}$ is used to generate
permutations in Algorithm~\ref{MainProtocol}.
Then, for any large enough $n \geq 100$, there exists a supermodular game with $n$ players such that
for any $\frac{1}{n} \leq \epsilon \leq 0.01$,
if the adversary has a perpetual violation budget of $1 \leq C \leq \epsilon n$,
then at least $\frac{n C}{10\epsilon}$ number of P-samples
is necessary to achieve $\epsilon$-expected maximin security. %\quan{Keep $n$, or use $\Gamma$}
\end{theorem}

\begin{proof}
We consider even $n = |\all|$, and define a game where player~$i^*$	is honest
and all other players are susceptible.
Suppose at most $R = \frac{Cn}{10 \epsilon}$ P-samples are used in
Algorithm~\ref{MainProtocol}.

\noindent \textbf{Game definition.} Let $Q \subseteq \all$ be a subset
of players such that $i^* \notin Q$ and $|Q| = \frac{n}{2}$.
Let $\alpha := \frac{2n(n-1)}{3n-2} \leq \frac{2n}{3}$.  %Consider game $G=(\all, v)$ such that $i^*$ get outcome $\frac{n(n-1)}{\frac{3}{2}n-1}$ if either the two following conditions is true:
	The utility function $v: 2^\all \rightarrow \R_+$ is defined as follows:
	
	$$
	v(S) = \begin{cases}
		2\alpha, & S=\all \\
		\alpha, & S = \all \setminus \{q \}, q \in \{i^*\} \cup Q\\
		0, & \text{otherwise.}
	\end{cases}
	$$
	
\noindent \textbf{Shapley Value $\phi_{i^*}$. } There are two cases when player~$i^*$ gets a non-zero reward:
(i) When player~$i^*$ is at the most preferable position, it gets a reward of $\alpha$; (ii) when~player $i^*$ is at the second most preferable position and a player $q \in Q$ is at the most preferable position, player $i^*$ also gets a reward of $\alpha$.
Therefore, the Shapley value of player $i^*$ is
    \begin{equation}
        \phi_{i^*} = \frac{1}{n}\cdot \alpha + \frac{1}{n}  \cdot \frac{n/2}{n-1}\cdot \alpha = 1 .
    \end{equation}

\noindent \textbf{Max-to-Mean Parameter $\Gamma$. }For any player $i \in \{i^*\} \cup Q$, employing the same argument as in the case when $i = i^*$, the Shapley value $\phi_i = 1$.
Since the game is supermodular, $U_{\max, i} = \mu_{i}(\all \backslash \{i\}) = \alpha$. Therefore, in this case, $\Gamma_i = \frac{U_{\max, i}}{\phi_i} = \alpha$. For any player $i \notin \{i^*\} \cup Q$, $i$ gets a non-zero reward if and only if $i$ is at the most favorable position. Therefore, in this case, $\Gamma_i = \frac{U_{\max, i}}{\phi_i} = \frac{U_{\max, i}}{\frac{1}{n}\cdot U_{\max, i}} = n$. Consequently, in this game, $\Gamma = \max(\alpha, n) = n$.
Next, it suffices to demonstrate that a minimum of $\frac{n C}{10\epsilon}$ P-samples is required to achieve $\epsilon$-expected maximin security.
	%To prove the lemma, we show that if $r = \Omega(n^2)$ then there exists some constant $\epsilon$ such that $\mathbb E_{r,x} < (1-\epsilon) r$.
	%Since $x = \Theta(n)$ and $v_{max} = \alpha = \Theta(n)$, the lemma can be proved.
	
\noindent \textbf{Adversarial Strategy.}  The $R$
iterations in Algorithm~\ref{MainProtocol} will be divided into $C$ blocks,
where each block corresponds to generating $T = \frac{R}{C} = \frac{n}{10 \epsilon}$ number of P-samples.
The adversary will instruct at most one susceptible player in $Q$ to
violate the protocol in each block.  Recall that under the perpetual violation model,
any detected violating player will be removed by the protocol.

\noindent \emph{Violation Opportunity.}
During the generation of a
P-sample in $\mathsf{SeqPerm}$, recall that $\mathsf{RandElim}$ is called repeatedly.
We say that there is a \emph{violation opportunity} when the
following event happens.
\begin{itemize}
\item When there are 3 players remaining
in the active set $S$,
it is of the form $S = \{i^*, q, y\}$
for some $q \in Q$ and $y \notin Q$.

\item When $\mathsf{RandElim}$ is run on $S$,
the adversary sees that~$y$ is about to be eliminated.
\end{itemize}

\noindent \emph{Adversarial Action.}
When a violation opportunity first arises within each block (of size $T$) as above,
the adversary instructs the player $q \in Q \cap S$ to abort.
Recall that the adversary causes violation at most once
in each block. This is for the ease of analysis. It is of course to the adversary's advantage to
seize every violation opportunity provided it has enough budget.

Observe that when a violation opportunity arises,
the honest player~$i^*$ will receive a reward of $\alpha$ if
the adversary does not interfere.
However, if the adversary instructs $q$ to abort,
then $i^*$ can only receive the reward if it beats $y$ in the last elimination round.
Therefore, by seizing a violation opportunity within one P-sample,
the expected damage to the honest player is $\frac{\alpha}{2}$.

\noindent \textbf{Probability of Violation in a Block.}
Suppose $s \leq \epsilon n$ violating players from $Q$ have already
been removed at the beginning of a block.
Then, the probability of a violation opportunity
happening in a P-sample is:

 $p_s := \frac{3!(|Q|-s)(n-|Q|-1)}{(n-s)(n-s-1)(n-s-2)} \cdot \frac{1}{3} = \frac{2(\frac{n}{2}-s)(\frac{n}{2}-1)}{(n-s)(n-s-1)(n-s-2)} \geq p_0 := \frac{n(\frac{n}{2}-1)}{n(n-1)(n-2)} \geq \frac{1}{n}(\frac{1}{2} - \frac{1}{n})$,

for $0 \leq s \leq 0.01n$ and $n \geq 100$.

Hence, the probability that
the adversary can seize a violation probability within a block of $T = \frac{n}{10 \epsilon}$ samples is at least:

$1 - (1 - p_0)^T \geq 1 - \exp(- p_0 \cdot T) \geq 0.9$, for $n \geq 100$ and $\frac{1}{n} \leq \epsilon \leq 0.01$.

\noindent \textbf{Bounding the Expectation Due to Removed Violating Players.}
Suppose $s \leq \epsilon n$ violating players have been removed.
Then, in each P-sample, the reward of the honest player~$i^*$ is:

$(\frac{1}{n-s} + \frac{1}{n-s} \cdot \frac{\frac{n}{2}-s}{(n-s-1)}) \cdot \alpha
\leq \frac{1}{1 - \epsilon} \cdot (1 + \frac{\frac{1}{2} - \epsilon}{1 - 2 \epsilon}) \cdot \frac{2}{3} \leq 1 + 2 \epsilon$,

when $\frac{1}{n} \leq \epsilon < 0.01$.

\noindent \textbf{Expected Reward within a Block of $T = \frac{n}{10 \epsilon}$ P-Samples.}
Under this adversarial attack strategy, the expected reward
for player~$i^*$ within a block is at most:

$(1 + 2 \epsilon) \cdot T - 0.9 \cdot \frac{\alpha}{2}
= (1 + 2\epsilon ) \cdot T
- 0.9 \cdot \frac{1}{3} \cdot n \cdot \frac{n - 1}{n - \frac{2}{3}}
< (1 - \epsilon) \cdot T$,
for $n \geq 100$.

Repeating the argument for each of the $C$ blocks
shows that $\epsilon$-expected maximin security is not achieved.
\end{proof}

\begin{remark}[General Conditions for Lower Bound]
\label{remark:LB}
Observe that in Theorem~\ref{LB_of_Level1},
the same asymptotic lower bound on the number of P-samples will hold,
as long as we use a permutation generation protocol with the following properties
in each P-sample:

\begin{itemize}
\item A violation opportunity happens with probability at least $\Omega(\frac{1}{n})$.

\item When a violation opportunity arises,
using at most $O(1)$ violation budget
is sufficient to cause a damage of $\Omega(n)$ times the Shapley value~$\phii$
of an honest player.
%\quan{Change to $\Gamma$?? Or it's not necessary?? Because the second condition make sure that $\Gamma = \Omega(n)$. Plus supermodular we obtain that $\Gamma = \Theta(n)$. Need further illustration? }
\end{itemize}
\end{remark}

%% file: adversary.tex
\section{Optimal Adversarial Strategy against $\mathsf{SeqPerm}$}
\label{sec:adversary}

In this section, we conduct a comprehensive analysis of the adversary's
optimal strategy against $\mathsf{SeqPerm}$ in supermodular games.
The approach is based on dynamic programming,
which computes the expected revenue of an honest player when faced with the adversary's optimal strategy.
As mentioned in Section~\ref{sec:prelim},
the number of entries in the dynamic program can be exponential in $n$.

%Additionally, we perform experiments to demonstrate that our algorithm provides significant protection for the expected revenue of the honest player, even when confronted with the adversary's optimal strategy.

\noindent \textbf{Adversary's Optimal Strategy.} Let $E_{worst}[T][S][c]$ represent the expected revenue of the honest player under the optimal strategy employed by the adversary, given that there are currently $S \subseteq N$ players participating in the protocol, $T$ permutations remaining (except the current one), and the adversary's budget is $c$.

	Let $i^*$ denote the honest player. In this round, there is a probability of $\frac{1}{|S|}$ that $i^*$ loses and obtains the $|S|^{th}$ preferable position in the permutation. In this scenario, the revenue of $i^*$ is determined, and there is no benefit for the adversary to abort. $i^*$ receives revenue equal to $\mu_{i^*}(\all \backslash S) + E_{\text{worst}}[T-1][\all][c]$.
However, if some player $i \neq i^*$ loses, then some player may abort to hurt $i^*$. If no one aborts, $i^*$ receives revenue equal to $E_{worst}[T][S\backslash \{i\}][c]$. However, if someone aborts, the adversary, following an optimal strategy, selects the player that can inflict the most harm on $i^*$. In such cases, $i^*$'s revenue is given by $\min_{j \in S\backslash \{i^*, i\}} E_{worst}[T][S\backslash \{j\}][c-1]$. Consequently, $i^*$'s expected revenue in this case is the minimum of these two terms.
%Let $\hat{j}$ represent the player whose abort would inflict the greatest harm on $i^*$, i.e., $E_{worst}[T][S\backslash \{\hat{j}\}][c-1] = \min_{j \in S\backslash \{i^*, i\}} E_{worst}[T][S\backslash \{j\}][c-1]$.
Let $\hat{j}$ represent the player whose abort would inflict the greatest harm on $i^*$, i.e., $\hat{j} = \arg\min_{j \in S\backslash \{i^*, i\}} E_{worst}[T][S\backslash \{j\}][c-1]$.
If the revenue corresponding to someone's abort is lower, in the optimal strategy of the adversary, $\hat{j}$ will abort in this case.

%From what we discussed above, given the DP-table, the optimal strategy of the adversary (to hurt $i^*$) is determined.

Based on the above discussion, $E_{worst}[T][S][c]$ can be represented using dynamic programming as follows:

\textbf{Base case:}
$$E_{worst}[0][\{i^*\}][0] = \mu_{i^*}(\all \backslash \{i^*\})$$

\textbf{Inductive case:}
\begin{equation}
	\label{adv_strategy}
	\begin{aligned}
		E_{\text{worst}}[T][S][c] = \frac{1}{|S|} \left(\mu_{i^*}(\all \backslash S) + E_{\text{worst}}[T-1][\all][c]\right) +  \\
		\frac{1}{|S|}\left(\sum_{i \in S \backslash \{i^*\}} \min\left(E_{\text{worst}}[T][S\backslash \{i\}][c], \min_{j \in S\backslash \{i^*, i\}} \left(E_{\text{worst}}[T][S\backslash \{j\}][c-1]\right)\right)\right)
	\end{aligned}
\end{equation}

\subsection{Implementation of Adversarial Strategy}

Note that when we employ the optimal adversarial strategy
in a simulation.  We need to start from large values of~$T$.
For instance, if the protocol uses $R$ number of P-samples,
then $E_{\text{worst}}[R-1]$ is used in the first P-sample.
Simply storing the entire DP table consumes an unreasonable amount of space.
We use two techniques to improve the efficiency of time and space.

\begin{enumerate}

\item We only store the values of $E_{worst}[T][\all][c]$ for each $0 \le T < R, 0 \le c \le C$.  In other words, we only store the entries
where $S = \all$.

When we use $E_{\text{worst}}$ for a specific $T$ in the simulation, we reconstruct $E_{\text{worst}}[T][S][c]$ for all $S \subseteq \all$ and $0 \le c \le C$ from $E_{\text{worst}}[T-1][\all][C]$ by formula (\ref{adv_strategy}).

\item  When we repeat the simulation $M$ times,
we perform them in parallel. For example, for the first P-sample,
when we reconstruct $E_{\text{worst}}[R-1]$,
we proceed with the first P-sample for $M$ simulations.
This is continued for subsequent P-samples.

In this way we only need to construct the entries of the DP table twice.
The first time is to store the values $E_{worst}[T][\all][c]$.
The second time is during reconstruction for parallel simulation.
\end{enumerate}

%% file: experiment.tex
\section{Experiments}
\label{sec:experiments}

In previous sections,
we have proved upper and lower bounds
on the number of P-samples to achieve several notions of maximin security.
In this section, we conduct experiments to empirically verify our theoretical results. The experiments are run on the computer with CPU i5-1240P, 1.70 GHz and a 16 GB RAM.
The source code can be found in our public repository\footnote{\url{https://github.com/l2l7l9p/coalitional-game-protocol-expcode}}.
To simulate the adversary's strategy,
we use a dynamic program approach to give
the optimal adversarial strategy to attack $\mathsf{SeqPerm}$,
where the details are given in \Cref{sec:adversary}.

\noindent \textbf{Game Settings.}
We use two supermodular games with different values of
the max-to-mean ratio $\Gamma$.

%\begin{itemize}

\noindent (1)  The first game is based on Theorem~\ref{LB_of_Level1},
which is synthetic and characterized by the number~$n$ of players.
Its definition is given in \Cref{sec:lowerbound}
and has $\Gamma$ close to $\frac 23 n$.

\noindent (2) The second is the edge synergy game
described in Section~\ref{sec:gamma}.
We use a real-world hypergraph that is
constructed from the collaboration network in DBLP,
where each person is a vertex and the co-authors in a publication
is a hyperedge.
For an honest player, we pick
an author, who in 2021, had 13 collaborators and 8 publications.
To simulate other susceptible players,
we include more players that represent scholars that
have no collaboration with the honest player,
finally resulting in a total of $n$ players.
In this case, $\Gamma = 3.15789$ and is independent of $n$.

%\end{itemize}

We evaluate the performance of Algorithm~\ref{MainProtocol} using
\texttt{SeqPerm} to generate random permutations.  Even though
the protocols can be run with large number of players,
the optimal adversarial strategy takes time exponential in~$n$.
Hence, we can only run experiments for small values of $n$ up to 200.

\ignore{
 Note that calculating the optimal strategy for the adversary simply by formula \ref{adv_strategy} is inefficient for a large $n$ because the number of states of $E_{worst}$ increases exponentially in $n$. To adapt it for these two games with large $n$, $S$ can be replaced by the number of players in $Q$ and out of $Q$ for the synthetic game, and by considering only members in the core for the hypergraph game.
}

\noindent \textbf{Expected Maximin Security.}
We perform experiments to verify Lemma \ref{lemma:ub_naiveperm} and Theorem \ref{LB_of_Level1}.
We evaluate the number $R$ of P-samples to reach $\epsilon$-expected maximin security under the optimal adversarial strategy.
We pick default values $n=100, C=20, \epsilon=0.01$,
and investigate how~$R$ varies with each of the three parameters.

%, which can be done by calculating formula \ref{adv_strategy}. Suggested by Theorem \ref{UB_of_level2} and \ref{LB_of_Level1}, the value should be related to the max-to-mean parameter $\Gamma$, the number of violation budget $C$ and the multiplicative error $\epsilon$. The number of players $n$ is linear to $\Gamma$ in the synthetic game and unrelated to $\Gamma$ in the hypergraph game, for which it is a good index to observe the effect of $\Gamma$. Therefore, we test the value of $R$ under the change of $n$, $C$ and $\epsilon$. By default we set $n=100, C=20, \epsilon=0.01$.

As predicted, Figure~\ref{exp_knownbudget} shows the results that $R$ is proportional to $C$ and inversely proportional to $\epsilon$.
In Figure~\ref{exp11_1},
we verify indeed that $R$ is proportional
to $\Gamma = \Theta(n)$ in the synthetic game;
in Figure~\ref{exp21_1},
we see that $R$ is independent of $n$,
because $\Gamma = O(1)$ for the edge synergy game.
Moreover, we see that $R \approx \frac 13 \cdot \frac{C\Gamma}{\epsilon}$ falls between its lower bound and upper bound
in the synthetic game.  On the other hand,
$R \approx 0.31 \cdot \frac{C\Gamma}{\epsilon}$ in the edge synergy game.

\begin{figure}
  \centering
  \begin{subfigure}{0.3\textwidth}
    \includegraphics[width=\textwidth]{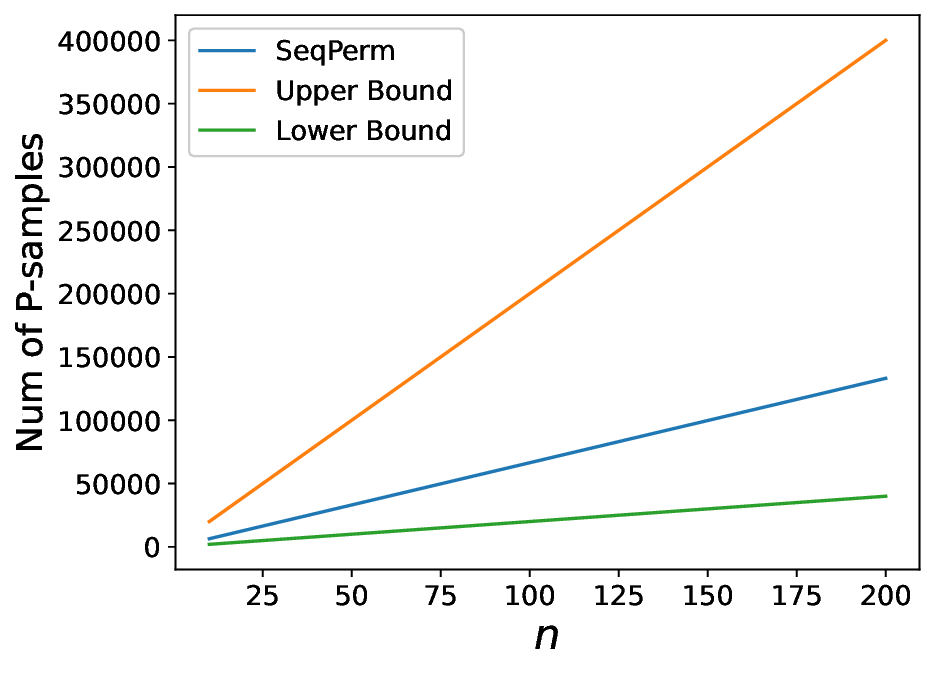}
    \caption{}\label{exp11_1}
	\end{subfigure}
  \begin{subfigure}{0.3\textwidth}
    \includegraphics[width=\textwidth]{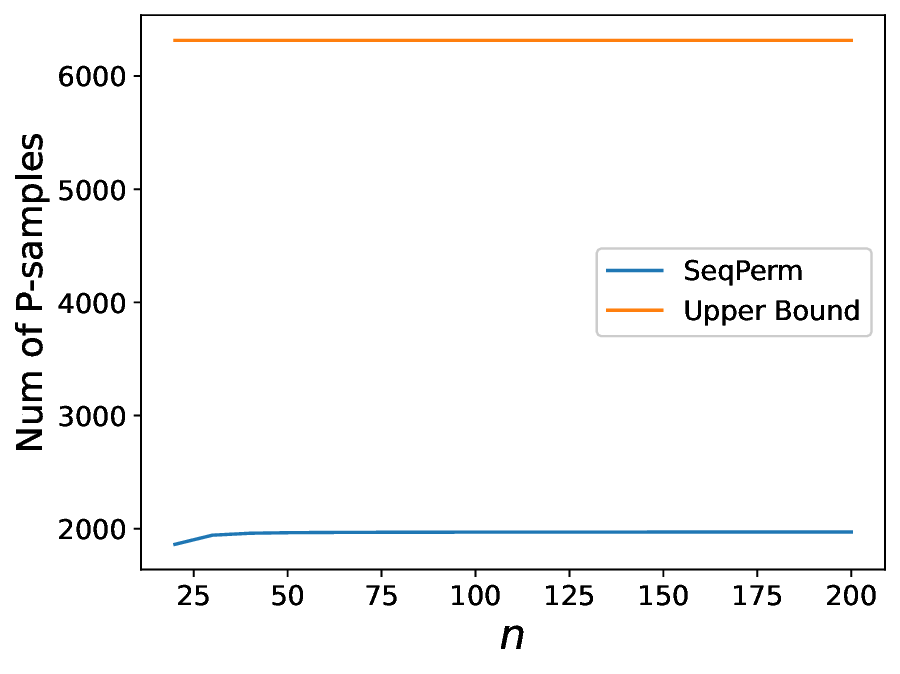}
    \caption{}\label{exp21_1}
  \end{subfigure}
  \begin{subfigure}{0.3\textwidth}
    \includegraphics[width=\textwidth]{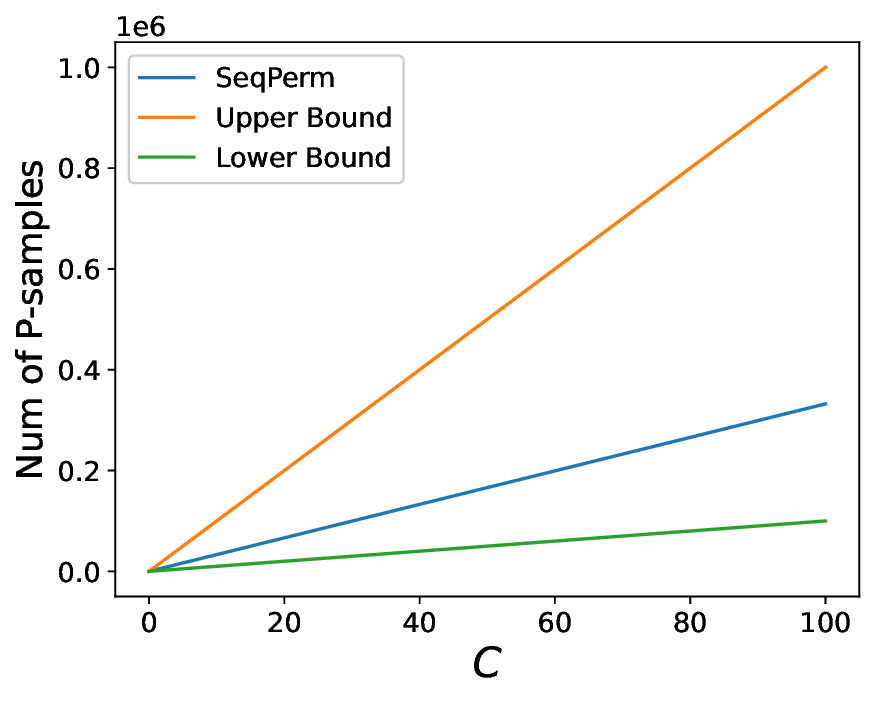}
    \caption{}\label{exp11_2}
  \end{subfigure}
	\begin{subfigure}{0.3\textwidth}
    \includegraphics[width=\textwidth]{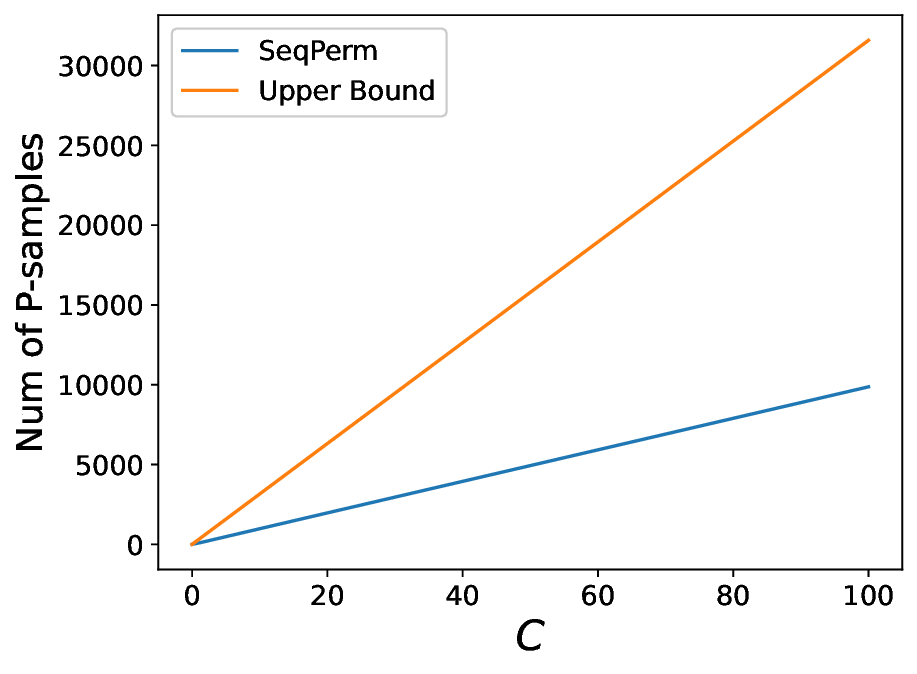}
    \caption{}\label{exp21_2}
  \end{subfigure}
  \begin{subfigure}{0.3\textwidth}
    \includegraphics[width=\textwidth]{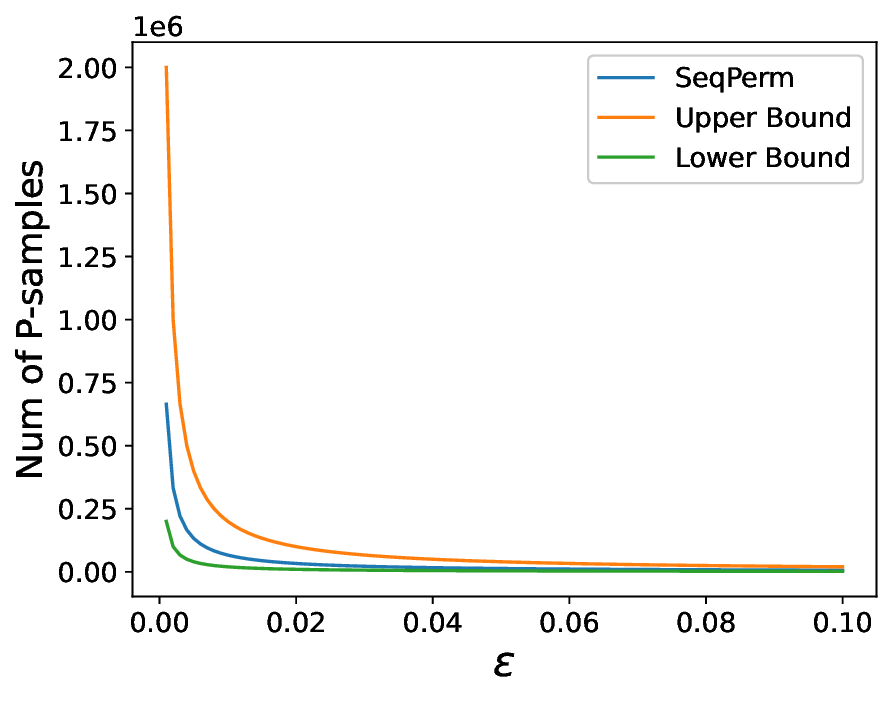}
    \caption{}\label{exp11_3}
  \end{subfigure}
	\begin{subfigure}{0.3\textwidth}
    \includegraphics[width=\textwidth]{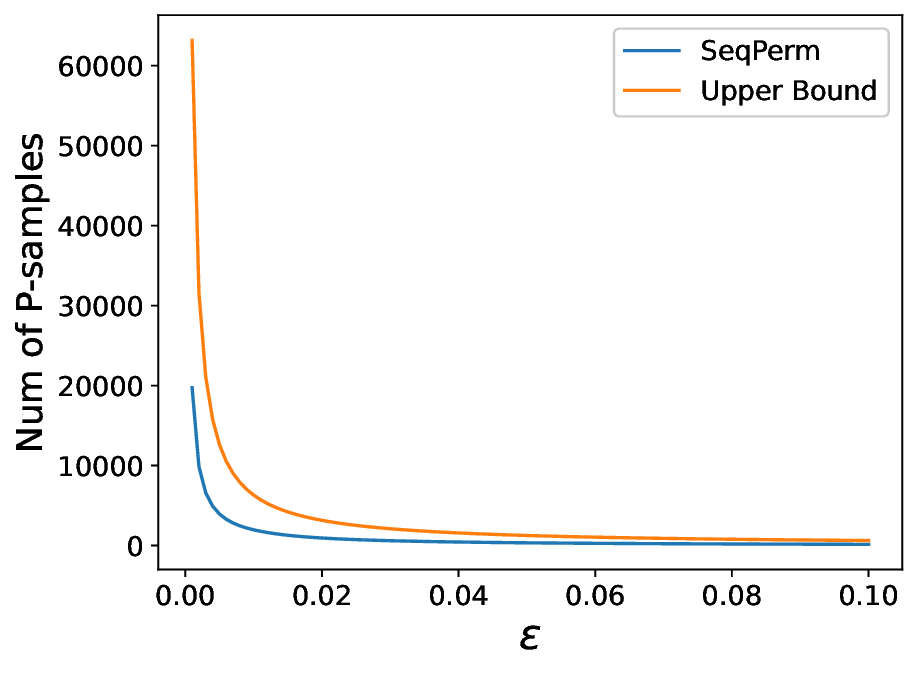}
    \caption{}\label{exp21_3}
  \end{subfigure}
  \caption{{\small The number $R$ of P-samples to reach $\epsilon$-expected maximin security under different $n$, $C$ and $\epsilon$. The default values are $n=100, C=20, \epsilon=0.01$. (\ref{exp11_1}), (\ref{exp11_2}) and (\ref{exp11_3}) are for the synthetic game; (\ref{exp21_1}), (\ref{exp21_2}) and (\ref{exp21_3}) are for edge synergy game on DBLP.}}\label{exp_knownbudget}
\end{figure}

\noindent \textbf{High Probability Maximin Security.}
We perform experiments to verify Theorem~\ref{UB_of_level2},
Corollary~\ref{cor:finite_unknown} and Theorem~\ref{lemma:adaptive_rate}.
Although the protocol descriptions are different,
one can check that, in each case,
the number of P-samples is $R \approx \max\{O(\frac{\Gamma}{\epsilon^2} \ln \frac{1}{\delta}), O(\frac{C \Gamma}{\epsilon})\}$.  Hence, we choose the parameters such that the two terms are the same.

\ignore{
In the case where the violation budget is unknown to public but known to the adversary, if $R_0$ defined in Lemma \ref{UB_of_unlimited} is no less than $\frac{2C\Gamma}{\epsilon}$, there is no need for the adversary to use any strategy because the protocol definitely terminates and achieve $(\epsilon, \delta)$-maximin security after $R_0$ P-samples, hence we only need to consider $R_0 \le \frac{2C\Gamma}{\epsilon}$ and for simplicity we let them equivalent. Then from the perspective of the adversary it is the same as the known-budget case where it can still use the strategy in Section \ref{adversary}. For the case where the violation budget is bounded by a fraction, the adversary can calculate the number of P-samples needed for Algorithm \ref{adprotocol} assuming $\epsilon_k$ halves until reaching the target $\epsilon$ and use the same strategy aiming at that number of P-samples. Even if the violation budget comes at the speed of the fraction, regarding it as having all the budget at the beginning makes little difference since the opportunities to violate appears uniformly. Therefore, the simulation with a fixed number of P-samples under the adversarial strategy in Section \ref{adversary} can be used to test both two cases.
}

For the synthetic game, we
set $n=100, C=200, R=8 \times 10^5, \epsilon=0.05, \delta=0.082$;
for the edge synergy game, we set
 $n=100, C=100, R=6316, \epsilon=0.1, \delta=0.082$.
%Parameters of the hypergraph game is a little smaller due to its inefficiency of the calculation of formula \ref{adv_strategy}.
We ran 1000 simulations for each game to plot
the empirical cumulative distribution of the
multiplicative error~$\widehat{\epsilon}$ for the honest player's Shapley value.

Figure \ref{exp_unknownbudget} shows the empirical cumulative distribution of the multiplicative error~$\widehat{\epsilon}$ in the 1000 simulations. The theoretical point $(\epsilon, 1-\delta)$ falls on the right of the curve, indicating that the simulated results are better than the theoretical bounds.

\begin{figure}
  \centering
  \begin{subfigure}{0.3\textwidth}
    \includegraphics[width=\textwidth]{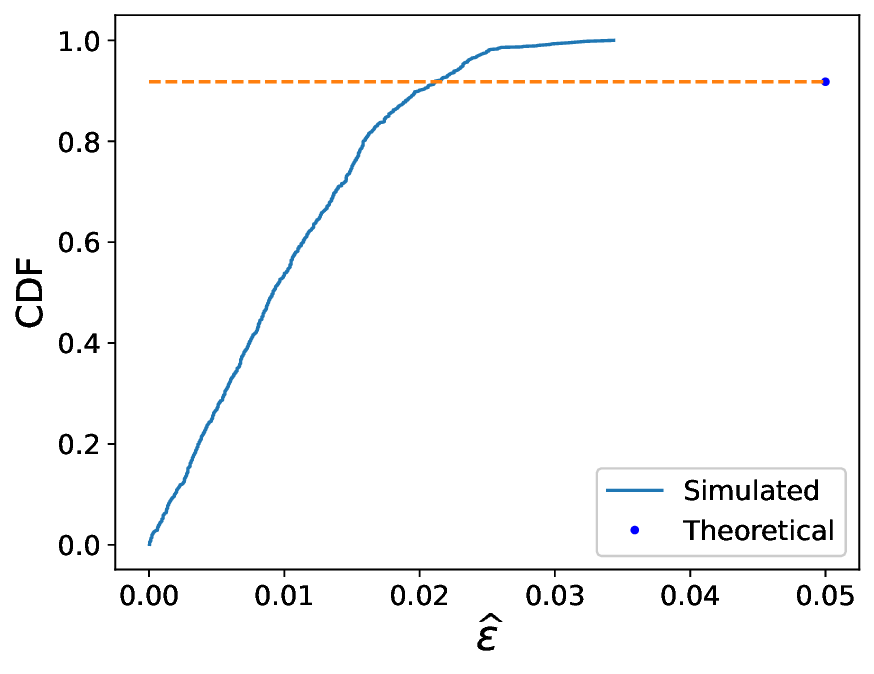}
    \caption{}\label{exp12}
  \end{subfigure}
  \begin{subfigure}{0.3\textwidth}
    \includegraphics[width=\textwidth]{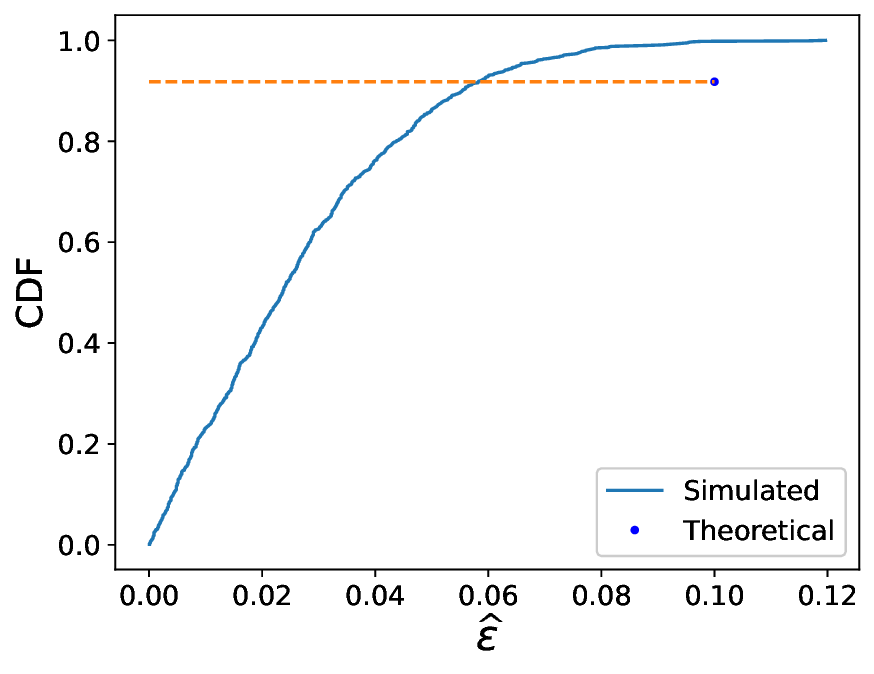}
    \caption{}\label{exp22}
  \end{subfigure}
  \caption{Empirical Cumulative Distribution Function of multiplicative error $\widehat{\epsilon}$ over 1000 simulations. (\ref{exp12}) is for the synthetic game, and (\ref{exp22}) is for the edge synergy game.}\label{exp_unknownbudget}
\end{figure} 

%% file: conclusion.tex
\section{Conclusion}
\label{sec:conclusion}

In this paper, we have designed game-theoretically secure distributed protocols
for coalitional games
which return allocations
 that approximate the Shapley value
with small constant multiplicative error.

By modeling the power of the adversary with a violation budget,
we have analyzed the upper bounds on the sampling complexity
for protocols that achieve high probability or expected maximin security.
Under the random permutation paradigm,
we also give a lower bound on the sampling complexity
for certain permutation generation protocols.  Our results
have been verified empirically with real-world data.

%is almost tight, and in the case of known and finite budget, the number of necessary permutations to achieve expected maximin security is also $\Theta(\frac{Cn}{\epsilon})$.

%Our results suggest that achieving high probability maximin security is more difficult than achieving expected maximin security, and the lower bound in Theorem \ref{LB_informal} applies to both versions of maximin security.

%We have provided a lower bound on the number of permutations based on \textsf{NaivePerm} and \textsf{SeqPerm}.

\ignore{

Here are interesting open questions related to this work.

%our lower bound can be applied to other types

%corresponding sampling complexity of other protocols that can generate permutations.

\begin{itemize}
	\item Under our adversarial setting, does there exist a distributed permutation generation protocol
	such that one (or a constant number of) P-sample is sufficient to guarantee expected maximin security?
	
	\item Even in the absence of an adversary, what are
	the tight upper and lower bounds on the sampling complexity
	of the utility function to achieve $\epsilon$-multiplicative error
	(for supermodular games) with high probability?
\end{itemize}
}